\documentclass[11pt,a4paper]{article}
\usepackage{amssymb}
\usepackage{amsmath}
\usepackage{amsfonts}
\usepackage{bm}
\usepackage{bbm}
\usepackage{amsthm}
\usepackage{mathrsfs}
\usepackage{hyperref}
\usepackage{color}
\usepackage{mathdots}
\usepackage[margin=2.41cm]{geometry}
\usepackage[all,cmtip]{xy}
\usepackage[utf8]{inputenc}
\usepackage{graphicx}
\usepackage{varwidth}
\usepackage{comment}
\usepackage{enumitem}
\usepackage{refcount}

\usepackage{mathtools}

\usepackage{upgreek}
\usepackage{rotating}

\usepackage{tikz}
\usetikzlibrary{shapes.geometric}
\usetikzlibrary{patterns.meta}
\usetikzlibrary{arrows.meta}
\usetikzlibrary{calc} 
\usetikzlibrary{cd}
\usetikzlibrary{matrix,calc}
\usetikzlibrary{decorations.pathreplacing,positioning}
\usetikzlibrary{decorations.markings,shapes.geometric,shapes.misc}    
\usetikzlibrary{graphs,decorations.pathmorphing,decorations.markings}

\tikzcdset{
  smooth squiggle/.style={
    decorate,
    decoration={
      snake,
      amplitude=0.45mm,
      segment length=2.2mm,
      pre length=0mm,
      post length=0.5mm
    },
    ->
  }
}

\definecolor{darkred}{rgb}{0.8,0.1,0.1}
\hypersetup{
	colorlinks=true,         
	urlcolor=darkred,
	linkcolor=darkred,
	citecolor=blue,
}

\theoremstyle{plain}
\newtheorem{theo}{Theorem}[section]
\newtheorem{lem}[theo]{Lemma}
\newtheorem{propo}[theo]{Proposition}
\newtheorem{cor}[theo]{Corollary}

\theoremstyle{definition}

\makeatletter
\newenvironment{theobis}[1]{%
  \renewcommand{\thetheo}{\getrefnumber{#1}$'$}%
  \begin{theo}
}{%
  \end{theo}
  \addtocounter{theo}{-1}
}
\makeatother

\newtheorem{exx}[theo]{Example}
\newenvironment{ex}
{\pushQED{\qed}\exx}
{\popQED\endexx}

\numberwithin{equation}{section}

\def\nn{\nonumber}

\def\bbC{\mathbb{C}}

\def\bbZ{\mathbb{Z}}

\def\d{{\rm d}}
\def\ii{{{\rm i}}}

\def\id{\mathrm{id}}
\def\supp{\mathrm{supp}}
\def\deg{\mathrm{deg}}
\def\dd{\mathrm{d}}

\def\cc{\mathrm{c}}

\def\1{I}

\def\ad{\mathrm{ad}}

\def\F{\mathcal{F}}

\def\L{\mathcal{L}}

\newcommand\wedgepair[2]{\left\langle #1 {}\stackrel{\scalebox{0.6}{$\wedge$}}{,}{} #2\right\rangle}
\newcommand{\pair}[2]{\langle\!\langle #1 , #2 \rangle\!\rangle}
\newcommand{\pairsize}[2]{\left\langle\!\left\langle #1 , #2 \right\rangle\!\right\rangle}

\DeclareMathAlphabet{\mathdutchcal}{U}{dutchcal}{m}{n}
\def\kay{\mathdutchcal{k}}

\def\g{\mathfrak{g}}

\def\CP{\mathbb{C}P^1}

\def\sk{\vspace{2mm}}

\makeatletter
\let\@fnsymbol\@alph
\makeatother

\title{%
Computational homological methods for integrable field theories
}

\author{%
Marco Benini$^{1,2,a}$, Ryan A.~Cullinan$^{3,b}$, Alexander Schenkel$^{4,c}$\ and\ Beno\^{\i}t Vicedo$^{3,d}$\vspace{4mm}\\
{\small ${}^1$ Dipartimento di Matematica, Dipartimento di Eccellenza 2023-27, Universit\`a di Genova,}\\
{\small Via Dodecaneso 35, 16146 Genova, Italy.}\vspace{2mm}\\
{\small ${}^2$ INFN, Sezione di Genova,}\\
{\small Via Dodecaneso 33, 16146 Genova, Italy.}\vspace{2mm}\\
{\small ${}^3$ Department of Mathematics, University of York,}\\
{\small Heslington, York YO10 5GH, United Kingdom.}\vspace{2mm}\\
{\small ${}^4$~Dipartimento di Matematica, Universit{\`a} di Trento and INFN-TIFPA,}\\
{\small Via Sommarive 14, 38123 Povo (Trento), Italy.}\vspace{4mm}\\
{\small \begin{tabular}{ll}
Email: & ${}^a$~\href{mailto:marco.benini@unige.it}{\texttt{marco.benini@unige.it}}\\
& ${}^b$~\href{mailto:ryan.cullinan@york.ac.uk}{\texttt{ryan.cullinan@york.ac.uk}}\\
& ${}^c$~\href{mailto:alexander.schenkel@unitn.it}{\texttt{alexander.schenkel@unitn.it}}\\
& ${}^d$~\href{mailto:benoit.vicedo@gmail.com}{\texttt{benoit.vicedo@gmail.com}}
\vspace{2mm}
\end{tabular}
}
}

\date{July 2026}

\begin{document}

\maketitle

\begin{abstract}
\noindent 
We develop explicit computational tools for the recent homological approach to the construction of $2$-dimensional integrable field theories on $\Sigma$ from $4$-dimensional semi-holomorphic Chern-Simons theory on $\Sigma \times C$. In this framework, the operation of integrating out the spectral curve $C$ is realized by homotopy transfer of a cyclic $L_\infty$-algebra associated with the $4$-dimensional theory with prescribed singularities and boundary conditions. We construct explicit strong deformation retracts for divisor-twisted Dolbeault complexes on $C=\mathbb{C}P^1$ and use them to make the transferred $L_\infty$-structure computationally accessible. As an application, we study the choice of meromorphic $1$-form corresponding to the principal chiral model with a Wess-Zumino term. We compute the transferred Maurer-Cartan action and the associated Lax connection, showing that the former resums to the standard principal chiral model action with a Wess-Zumino term and that the latter reproduces the usual Lax connection.
\end{abstract}
\vspace{-1mm}

\vspace{-2mm}

\vspace{-2mm}

\renewcommand{\baselinestretch}{0.8}\normalsize
\setcounter{tocdepth}{1}
\tableofcontents
\renewcommand{\baselinestretch}{1.0}\normalsize


\section{\label{sec:intro}Introduction and summary}
A remarkable feature of the gauge-theoretic approach to integrability based on $4$-dimensional
semi-holomorphic Chern-Simons theory, originally proposed in Nekrasov's thesis \cite{Nekrasov}
and later independently rediscovered and further developed by Costello, Witten and Yamazaki
\cite{Costello:2013zra,Costello:2013sla,Witten:2016spx,Costello:2017dso,Costello:2018gyb,CY3},
is that the geometric origin of the spectral parameter is made manifest.
More concretely, in the application to $2$-dimensional integrable field theories pioneered by
Costello and Yamazaki \cite{CY3}, the $2$-dimensional spacetime $\Sigma$ is extended to the
$4$-dimensional product manifold $X=\Sigma\times C$, on which semi-holomorphic Chern-Simons
theory is formulated, with the Riemann surface $C$ carrying the complex spectral parameter $z$
of the Lax connection.
The passage from this $4$-dimensional theory on $X$ to the $2$-dimensional integrable field theory on 
$\Sigma$ is usually described as ``integrating out $C$''. Through this mechanism, 
a remarkably broad family of $2$-dimensional integrable field theories is organized 
in terms of the input data of the associated $4$-dimensional theory, consisting of 
1.)~a structure group $G$, 2.)~a Riemann surface $C$, 3.)~a meromorphic $1$-form $\omega$ 
on $C$, and 4.)~prescribed boundary conditions and singularities for the gauge field $A$ 
at the poles and zeros of $\omega$, respectively. In this paper, we focus on the 
Riemann sphere $C=\CP$ and therefore specialize to this setting throughout the introduction.
\sk

Classically, ``integrating out $C = \CP$'' is understood to mean solving part 
of the bulk equations of motion in the holomorphic gauge $A_{\bar z} = 0$,
see \cite{CY3} as well as \cite{DLMV,Lacroix,BSV} for further developments. 
Let us briefly recall how this procedure works. The starting point is to write 
the $4$-dimensional gauge field $A$ as a ``formal'' gauge transformation \cite{CY3,DLMV,Lacroix}
\begin{equation} \label{A formal gauge}
A = - (\bar\partial + \d_\Sigma) \widehat{g} \, \widehat{g}^{-1} + \widehat{g} \, L \, \widehat{g}^{-1},
\end{equation}
where $\widehat{g}$ is a smooth $G$-valued function on $X$ while $L$ is a 
$\g$-valued $1$-form with components only along $\Sigma$ that have the same 
singularities as the $\Sigma$-components of $A$ at the zeros of $\omega$. 
Importantly, however, $L$ is not required to satisfy the boundary conditions 
imposed at the poles of $\omega$. It is in this sense that the gauge 
transformation \eqref{A formal gauge} is only ``formal'' since it does not 
preserve the space of allowed fields of the $4$-dimensional theory. Upon 
enforcing two of the bulk equations of motion for $A$, the field $L$ becomes 
meromorphic along $C$ with poles at the zeros of $\omega$, and is identified 
with the Lax connection of the $2$-dimensional integrable field theory. 
Substituting \eqref{A formal gauge} into the $4$-dimensional action produces 
the action of the $2$-dimensional integrable field theory, whose dynamics 
is given by the remaining bulk equation of motion encoding the flatness of $L$.
\sk

Recently, a mathematically rigorous homological refinement of this construction 
was developed in \cite{BSV2}, which conceptually clarifies the mechanism recalled 
above by which $4$-dimensional semi-holomorphic Chern-Simons theory on 
$X=\Sigma\times C$ gives rise to $2$-dimensional integrable field theories on 
$\Sigma$. In particular, it gives the operation of ``integrating out $C$'' 
a precise cohomological meaning and provides a clear interpretation of the 
``formal'' gauge transformation \eqref{A formal gauge} used in the construction 
of the Lax connection. The results of \cite{BSV2} were formulated and proved 
in the homological framework of (cyclic) $L_\infty$-algebras which is used to 
perturbatively encode action functionals and gauge symmetries of classical field 
theories. In the literature this approach is also known as the Batalin-Vilkovisky formalism, 
see e.g.\ \cite{CG1,CG2} for a modern perspective and also \cite{JRSW} for a review. 
In order to explain the goal of the present paper it is useful to first give a 
brief overview of the construction from \cite{BSV2}, focusing on those aspects relevant for 
the comparison with the standard approach recalled in the previous paragraph.
\sk

The starting point of \cite{BSV2} is the identification of the 
cyclic $L_\infty$-algebra $\big(\F(X),\ell\big)$ corresponding to $4$-dimensional 
semi-holomorphic Chern-Simons theory on $X = \Sigma \times C$. Since the theory is 
topological along $\Sigma$ and holomorphic along $C = \CP$, the underlying cochain complex 
is modeled on the complex of de Rham-Dolbeault forms on $X$ with differential 
$\d_\Sigma + \bar\partial$ given by the sum of the de Rham differential $\d_\Sigma$ 
along $\Sigma$ and the Dolbeault differential $\bar\partial$ along $C$. However,
to allow singularities and enforce boundary conditions at the zeros and 
poles of $\omega$, the 
various Dolbeault forms entering $\F(X)$ are twisted by suitable 
divisors. For instance, a generic field $A \in \F(X)^1$ in degree $1$ is
given by a $\g$-valued $1$-form $A = A_\Sigma + A_{\bar z}$ 
whose components display the following distinctive features: The components 
$A_\Sigma$ along $\Sigma$ have certain singularities determined by a splitting of the 
zeros of $\omega$, the $(0,1)$-component $A_{\bar z}$ along $C$ is 
instead smooth, and both $A_\Sigma$ and $A_{\bar z}$
are subject to certain boundary conditions at the poles of $\omega$. 
The resulting cochain complex $\big( \F(X), \d_\Sigma + \bar\partial \big)$
carries a natural Lie bracket $\ell_2$ making it into a differential graded Lie algebra, which we regard as
an $L_\infty$-algebra $\big( \F(X), \ell \big)$ with $L_\infty$-structure $\ell = (\ell_n)_{n \geq 1}$ given by
the differential $\ell_1 = \d_\Sigma + \bar\partial$, the Lie bracket $\ell_2$
and vanishing higher brackets $\ell_n = 0$, for $n \geq 3$. Furthermore, 
the twists of the various Dolbeault forms are 
chosen so that $\big( \F(X), \ell \big)$ admits also a cyclic structure. The Maurer-Cartan
action functional associated with the resulting cyclic $L_\infty$-algebra $\big( \F(X), \ell \big)$ coincides
with the action functional of $4$-dimensional semi-holomorphic Chern-Simons theory from \cite{CY3}.
\sk

In this language, the operation of ``integrating out $C$'' can be given a rigorous 
homological meaning \cite{BSV2}: It is realized by homotopy transfer of the cyclic 
$L_\infty$-algebra $\big( \mathcal F(X), \ell \big)$ along a strong deformation retract 
to the partial Dolbeault cohomology $\F(\Sigma) = \mathsf H^\bullet \big( \F(X), \bar\partial \big)$ along $C$. The upshot of this construction is a 
cyclic $L_\infty$-algebra $\big( \F(\Sigma), \ell' \big)$
 which defines a 
(perturbative) $2$-dimensional field theory on $\Sigma$. The homotopy transfer 
theorem also provides an equivalence of $L_\infty$-algebras
\begin{equation} \label{infty morph i intro}
\begin{tikzcd}
\widetilde{i}_\infty \; : \; \big( \F(\Sigma), \ell' \big) \arrow[r, smooth squiggle, "\sim"] & \big( \F(X), \ell \big),
\end{tikzcd}
\end{equation}
which formally establishes the classical equivalence between the $4$-dimensional 
semi-holomorphic Chern-Simons theory on $X$ and the associated $2$-dimensional field theory on $\Sigma$.
\sk

While this construction gives a conceptual and mathematically precise account of the 
passage from the $4$-dimensional semi-holomorphic Chern-Simons theory on $X$ to the 
$2$-dimensional integrable field theory on $\Sigma$, the transferred cyclic 
$L_\infty$-structure $\ell'$ on $\F(\Sigma)$ was only defined implicitly in 
\cite{BSV2}, with no computationally accessible presentation.
Indeed, the homotopy transfer construction from \cite{BSV2} involves the choice of strong deformation retracts
\begin{equation} \label{SDRs intro}
\begin{tikzcd}
\mathsf H^\bullet\Omega^{0,\bullet}(C,L_D)
 \arrow[r,"i_D",shift left=1mm]
&
\Omega^{0,\bullet}(C,L_D)
 \arrow[l,"p_D",shift left=1mm]
 \arrow["h_D"',out=-15,in=15,distance=8mm]
\end{tikzcd}
\end{equation}
for each divisor-twisted Dolbeault complex appearing in 
$\F(X)$. Although Hodge theory ensures the existence of such 
data, it does not provide a computationally accessible way 
of evaluating the transferred cyclic $L_\infty$-structure $\ell'$ 
or the $\infty$-quasi-isomorphism $\widetilde{i}_\infty$ in \eqref{infty morph i intro}. One of the central 
aims of the present paper will be to remedy this shortfall by explicitly constructing 
strong deformation retracts as in \eqref{SDRs intro} for divisors $D : C \to \bbZ$ 
of any degree, see Section \ref{sec:explicitSDR},
thus paving the way for the explicit computation of the transferred Maurer-Cartan action functional and associated Lax connection.
\sk

The construction of the Lax connection associated with the above $2$-dimensional 
field theory, modeled by the transferred cyclic $L_\infty$-algebra $\big( \F(\Sigma), \ell' \big)$,
also admits a conceptual and transparent homological interpretation in the approach of 
\cite{BSV2}. Specifically, one can define a differential graded Lie algebra, which 
we regard as an $L_\infty$-algebra $\big(\L(X),\ell\big)$, in 
much the same way as $\big(\F(X),\ell\big)$, where one allows singularities only, while no boundary condition is imposed.
In other words, although $\L(X)$ is modeled again on the complex of de 
Rham-Dolbeault forms on $X$, the various Dolbeault forms are 
twisted exclusively by non-negative divisors, which are related to the zero divisor of 
$\omega$, to allow for specified singularities at the zeros of $\omega$. For instance, a generic field 
$\widetilde{A} \in \L(X)^1$ in degree $1$ is given by a $\g$-valued $1$-form 
$\widetilde{A} = \widetilde{A}_\Sigma + \widetilde{A}_{\bar z}$ whose components display 
the following distinctive features: The components 
$\widetilde{A}_\Sigma$ along $\Sigma$ have certain singularities specified by a
splitting of the zeros of $\omega$, the $(0,1)$-component $\widetilde{A}_{\bar z}$ 
along $C$ is instead smooth, and neither $\widetilde{A}_\Sigma$ nor $\widetilde{A}_{\bar z}$
are subject to any boundary condition. The same homotopy transfer construction as above produces out of the 
$L_\infty$-algebra $\big( \L(X), \ell \big)$ an equivalent  
$L_\infty$-algebra $\big(\L(\Sigma),\ell'\big)$, which in this specific case turns out to be a differential
graded Lie algebra, describing meromorphic flat 
connections along $\Sigma$. The equivalence is witnessed by an $\infty$-quasi-isomorphism,
which in this specific case turns out to be a quasi-isomorphism of differential graded Lie algebras
\begin{equation} \label{morph i intro}
\begin{tikzcd}
i \; : \; \big( \L(\Sigma), \ell' \big) \arrow[r, "\sim"] & \big( \L(X), \ell \big),
\end{tikzcd}
\end{equation}
whose quasi-inverse is an $\infty$-quasi-isomorphism
\begin{equation} \label{infty morph p intro}
\begin{tikzcd}
\widetilde{p}_\infty \; : \; \big( \L(X), \ell \big) \arrow[r, smooth squiggle, "\sim"] & \big( \L(\Sigma), \ell' \big).
\end{tikzcd}
\end{equation}
With this additional structure, the assignment of the Lax connection can now be described as 
follows: Define an $\infty$-morphism via the composition
\begin{equation} \label{lambda infty def intro}
\begin{tikzcd}[column sep=12mm, row sep=10mm]
(\mathcal F(\Sigma), \ell') \arrow[d, smooth squiggle, "\sim", "\widetilde{i}_\infty"'] \arrow[r, smooth squiggle, "\lambda_\infty"] & (\mathcal L(\Sigma), \ell')\\
(\mathcal F(X), \ell) \arrow[r, "\psi"'] & (\mathcal L(X), \ell) \arrow[u, smooth squiggle, "\sim"', "\widetilde{p}_\infty"]
\end{tikzcd}
\end{equation}
where the bottom morphism is defined by 
forgetting the boundary conditions at the poles of $\omega$. The resulting $\infty$-morphism 
$\lambda_\infty$ in \eqref{lambda infty def intro} then assigns to each on-shell 
field configuration of the $2$-dimensional field theory on $\Sigma$, i.e.\ to each 
Maurer-Cartan element in $\big( \mathcal F(\Sigma), \ell' \big)$, its flat Lax 
connection, i.e.\ a Maurer-Cartan element in $\big( \mathcal L(\Sigma), \ell' \big)$.
\sk

The main aim of the present paper is to render the abstract homological 
framework of \cite{BSV2} explicit and computationally accessible. To this end, 
in Section \ref{sec:explicitSDR} we construct explicit strong deformation retracts 
for divisor-twisted Dolbeault complexes as in \eqref{SDRs intro}. 
These enable us to explicitly compute both the Maurer-Cartan action functional $S_{\rm MC}$, associated with the transferred 
cyclic $L_\infty$-algebra $\big(\mathcal F(\Sigma),\ell'\big)$ encoding the resulting 
$2$-dimensional field theory, and the corresponding Lax connection. For simplicity 
we will focus on a particular choice of meromorphic $1$-form $\omega$, which corresponds 
to the principal chiral model with a Wess-Zumino term (in short, WZ-term). We show in Section \ref{sec: deriving MC action}
that in this case the Maurer-Cartan action functional $S_{\rm MC}[\phi]$ evaluated on a field configuration 
$\phi \in \F(\Sigma)^1 = \Omega^0(\Sigma, \g)$ resums to the action functional of the 
principal chiral model with a WZ-term for the group-valued field given by the formal 
exponential $g = e^\phi$. Furthermore, in Section \ref{sec:Lax connection} we show that, up to a gauge transformation, 
$\lambda_\infty(\phi)$ coincides with the usual Lax connection 
$L \in \L(\Sigma)^1$ of the principal chiral model with a WZ-term. More specifically, we show that
\begin{equation} \label{A formal gauge precise}
\psi(A) = - (\d_\Sigma + \bar\partial) g_2 g_2^{-1} + g_2 \,i(L)\, g_2^{-1},
\end{equation}
where $A = \widetilde{i}_\infty(\phi) \in \F(X)^1$ is the image of the field configuration
$\phi \in \F(\Sigma)^1$ under the $\infty$-quasi-isomorphism \eqref{infty morph i intro} 
and $i(L) \in \L(X)^1$ is the image of the Lax connection $L \in \L(\Sigma)^1$ under the 
quasi-isomorphism \eqref{morph i intro}. The honest gauge transformation in \eqref{A formal gauge precise} should 
be compared with the ``formal'' gauge transformation in \eqref{A formal gauge} relating the 
singular gauge field $A \in \F(X)^1$ satisfying boundary conditions to the meromorphic 
connection $L \in \L(\Sigma)^1$, or its canonical inclusion $i(L) \in \L(X)^1$ in 
$\L(X)$, where singular gauge fields are not subject to boundary conditions.
After applying the additional map $\psi$, which forgets the boundary conditions, 
we see that this ``formal'' gauge transformation becomes an honest gauge transformation 
in the $L_\infty$-algebra $\big( \L(X), \ell \big)$.
\sk

Although we focus on the principal chiral model with a WZ-term as a concrete proof of 
concept, we expect the computational methods developed here to apply more broadly. 
In particular, the computations of Section \ref{sec:application} should extend to more 
general choices of meromorphic $1$-form $\omega$ on $C=\CP$, as well as to other 
choices of singularities and local boundary conditions in the $L_\infty$-algebra 
$\big(\F(X),\ell\big)$ of $4$-dimensional semi-holomorphic Chern-Simons theory. Our treatment 
of the principal chiral model with a WZ-term should also be compared with the recent work 
\cite{Alfonsi:2026pth}, which develops a complementary homotopy-algebraic approach to the 
principal chiral model. In that work, cyclic $L_\infty$-algebras are constructed for both 
$4$-dimensional semi-holomorphic Chern-Simons theory and the principal chiral model, and an 
explicit $\infty$-quasi-isomorphism between them is constructed using the known Lax connection 
of the principal chiral model as input. By contrast, in the present work the Lax connection 
of the principal chiral model with a WZ-term is derived by explicitly evaluating the $\infty$-quasi-isomorphism $\widetilde{i}_\infty$ in \eqref{infty morph i intro}, 
determined by the homotopy transfer theorem, on the $2$-dimensional field $\phi \in \F(\Sigma)^1$. 
This allows us to obtain a closed-form expression for the $4$-dimensional gauge field 
$A = \widetilde{i}_\infty(\phi) \in \F(X)^1$ from which the Lax connection can be read off 
after applying the morphism $\psi$, which forgets the boundary conditions, as in 
\eqref{A formal gauge precise}. Moreover, the action of the resulting $2$-dimensional 
integrable field theory is itself derived by homotopy transfer of the cyclic $L_\infty$-algebra 
$\big(\F(X),\ell\big)$. More generally, we expect an important advantage of our constructive 
approach to emerge in applications to higher-dimensional topological-holomorphic Chern-Simons 
theories, such as those considered from a homological perspective in \cite{BCSV}, where the 
direct generalizations of standard techniques appear to be more limited in scope, see \cite{SV,CL}.
\sk

Let us briefly outline the structure of the paper. In Section \ref{sec:prelim} we recall 
the necessary homological background and review the relevant results of \cite{BSV2}. In 
Section \ref{sec:explicitSDR} we construct explicit strong deformation retracts for 
divisor-twisted Dolbeault complexes on $\CP$, which play a pivotal role in all subsequent computations. 
In Section \ref{sec:application} we use these strong deformation retracts to make the framework of \cite{BSV2} 
explicit in a specific example, where by computing both the 
transferred Maurer-Cartan action functional and the associated Lax connection
we recognize the principal chiral model with a WZ-term. The technical verification of the 
strong deformation retract identities from Section \ref{sec:explicitSDR} is deferred to 
Appendix \ref{sec:proof of SDR}.


\section{\label{sec:prelim}Preliminaries}
In this section we recall the main results of \cite{BSV2}, which provide a homological 
perspective on the Costello-Yamazaki construction of $2$-dimensional integrable field 
theories from $4$-dimensional semi-holomorphic Chern-Simons theory.

\subsection{Perturbative classical field theories and $L_\infty$-algebras} \label{sec:pertFT and Linfty}
From a field-theoretic point of view, $L_\infty$-algebras provide an efficient and 
systematic way of perturbatively encoding the equations of motion of a classical field 
theory and its gauge symmetries. Specifically, an $L_\infty$-algebra is a pair $(L, \ell)$ 
consisting of a $\bbZ$-graded vector space $L = (L^i)_{i \in \bbZ}$ and a family 
$\ell = (\ell_n)_{n \geq 1}$ of graded antisymmetric linear maps
\begin{equation}
\ell_n : L^{\otimes n} \longrightarrow L
\end{equation}
of degree $|\ell_n| = 2-n$. These maps are required to satisfy the homotopy Jacobi identities
\begin{equation} \label{L-infty algebra relations}
\sum_{k+l-1 = n} \sum_{\sigma \in \text{Sh}(l, k-1)} (-1)^{|\sigma|} (-1)^{k-1} \ell_k \circ \Big( \ell_l \otimes \id^{\otimes (k-1)} \Big) \circ \gamma_\sigma = 0
\end{equation}
for every $n \geq 1$, where $\text{Sh}(l,l^\prime) \subset \Sigma_n$ is the set 
of $(l,l^\prime)$-shuffle permutations, and for a permutation $\sigma \in \Sigma_n$ 
we denote by $(-1)^{|\sigma|}$ its parity and by $\gamma_\sigma : L^{\otimes n} \to L^{\otimes n}$ 
its action on $L^{\otimes n}$ induced by the symmetric braiding on graded vector spaces given 
by $x \otimes y \mapsto (-1)^{|x| |y|} y \otimes x$, for homogeneous $x, y \in L$. Note 
that for $n=1$ the relation \eqref{L-infty algebra relations} reduces to the condition
$\ell_1 \circ \ell_1 = 0$, so that every $L_\infty$-algebra $(L, \ell)$ has an underlying 
cochain complex $(L, \ell_1)$.
\sk

In the convention we use here, the fields of the associated classical field theory 
are degree $1$ elements $\alpha \in L^1$ whose field equations are given by the Maurer-Cartan equation
\begin{subequations} \label{MC prelim}
\begin{equation} \label{MC eq}
\sum_{n\geq 1} \frac{1}{n!} \, \ell_n(\alpha^{\otimes n}) = 0 .
\end{equation}
If $L^0$ is non-trivial then degree $0$ elements $\epsilon \in L^0$ parametrize 
infinitesimal gauge transformations, under which a field $\alpha \in L^1$ transforms as
\begin{equation} \label{MC gauge tr}
\delta_\epsilon \alpha = \sum_{n \geq 0} \frac{1}{n!} \, \ell_{n+1}(\alpha^{\otimes n}, \epsilon) .
\end{equation}
Similarly, the elements in negative degrees $L^i$ for $i < 0$ parametrize 
higher gauge transformations. When the $L_\infty$-algebra is equipped with a compatible 
cyclic structure, namely a degree $-3$ graded symmetric linear map 
$\langle\!\langle \cdot, \cdot \rangle\!\rangle : L \otimes L \to \bbC$ invariant 
under the $L_\infty$-structure, the action functional of the classical field theory is given by
\begin{equation} \label{MC action}
S[\alpha] = \sum_{n\geq 1} \frac{1}{(n+1)!} \big\langle\!\big\langle \alpha, \ell_n(\alpha^{\otimes n}) \big\rangle\!\big\rangle \quad .
\end{equation}
\end{subequations}
The Euler-Lagrange equations of the action are precisely the Maurer--Cartan equation \eqref{MC eq}.
\sk

Importantly, the $L_\infty$-algebra relations \eqref{L-infty algebra relations} 
ensure the consistency of the equations of motion \eqref{MC eq} in the form of generalized 
Bianchi identities, their covariance under gauge transformations \eqref{MC gauge tr}, 
and the closure of the gauge algebra up to higher transformations and possibly equations of motion.
One should view the (possibly infinitely many) terms in the sums in \eqref{MC eq}, 
\eqref{MC gauge tr} and \eqref{MC action} as the terms in the perturbative expansion 
of the equations of motion, infinitesimal gauge transformation and action functional, 
respectively. The cyclic $L_\infty$-algebra framework therefore provides a perturbative 
packaging of the field equations, (higher) gauge transformations and action functional 
of a classical field theory into a single algebraic structure.
\begin{ex} \label{ex: 3d CS}
A key example of the $L_\infty$-algebra framework, closely related to the main 
example considered in \cite{BSV2} to be recalled below, is that of Chern-Simons 
theory on a $3$-dimensional oriented manifold $M$. Let $\g$ be a fixed Lie algebra 
with Lie bracket $[\cdot, \cdot] : \g \otimes \g \to \g$ and equipped with a non-degenerate 
invariant bilinear form $\langle \cdot, \cdot \rangle : \g \otimes \g \to \bbC$. The 
underlying cochain complex $(L, \ell_1)$ is the $\g$-valued de Rham complex on $M$, namely
\begin{equation}
\big(\mathcal F_{\rm CS}(M), \d\big) \; \coloneqq \; \Bigg(
\begin{tikzcd}
\overset{(0)}{\Omega^0(M, \g)} \arrow[r, "\d_M"] & \overset{(1)}{\Omega^1(M, \g)} \arrow[r, "\d_M"] & \overset{(2)}{\Omega^2(M, \g)} \arrow[r, "\d_M"] & \overset{(3)}{\Omega^3(M, \g)}
\end{tikzcd}
\Bigg).
\end{equation}
The $L_\infty$-algebra $(L, \ell)$ has a non-trivial $\ell_2$-bracket 
$\ell_2 : \mathcal F_{\rm CS}(M)^{\otimes 2} \to \mathcal F_{\rm CS}(M)$, 
$\alpha \otimes \beta \mapsto [\alpha \,\overset{\wedge},\, \beta]$ given by 
combining the Lie bracket on $\g$ with the wedge product on forms, but all of 
the higher $\ell_n$-brackets for $n \geq 3$ are trivial $\ell_n = 0$. In this 
case we say that $\ell = (\ell_1, \ell_2)$ defines a differential graded Lie algebra 
structure on $L$. The cyclic structure is defined only on compactly supported fields and given by
\begin{equation}
\langle\!\langle \cdot, \cdot \rangle\!\rangle \; : \; \mathcal F_{{\rm CS}, c}(M) \otimes \mathcal F_{{\rm CS}, c}(M) \;\longrightarrow\; \bbC, \qquad \alpha \otimes \beta \; \longmapsto \; \int_M \langle \alpha \, \overset{\wedge}, \, \beta \rangle .
\end{equation}
The field content is given by a $\g$-valued $1$-form $A \in \mathcal F_{\rm CS}(M)^1 = \Omega^1(M, \g)$ 
and the Maurer-Cartan action \eqref{MC action} for a compactly supported field coincides with the Chern-Simons action
\begin{equation}
S[A] \;=\; \int_M \big\langle A \, \overset{\wedge}, \, \tfrac 12 \d_M A + \tfrac{1}{3!} [A \, \overset{\wedge}, \, A] \big\rangle .
\end{equation}
The corresponding equation of motion \eqref{MC eq} is the flatness condition 
$\d_M A + \tfrac 12 [A \,\overset{\wedge},\, A] = 0$ and \eqref{MC gauge tr} recovers 
the usual infinitesimal gauge transformation formula 
$\delta_\epsilon A = \d_M \epsilon + [A \, \overset{\wedge}, \, \epsilon]$ 
with parameter $\epsilon \in \mathcal F_{\rm CS}(M)^0 = \Omega^0(M, \g)$.
\end{ex}

\subsection{$4$-dimensional semi-holomorphic Chern-Simons theory}
The $4$-dimensional semi-holomorphic Chern-Simons theory of Costello and Yamazaki 
\cite{CY3} is a topological-holomorphic analogue of Chern-Simons theory defined 
on the $4$-dimensional product manifold
\begin{equation}
X \coloneqq \Sigma \times C
\end{equation}
with $\Sigma$ a $2$-dimensional Lorentzian spacetime and $C$ a compact Riemann 
surface. For simplicity, throughout this paper we will only be concerned with 
the case of the Riemann sphere $C = \CP$, but we shall keep denoting it by $C$ for brevity.
We shall be interested in two separate $L_\infty$-algebras associated with this theory, 
constructed in \cite{BSV2}. To motivate their definition, it is useful to recall the 
form of the action of $4$-dimensional semi-holomorphic Chern-Simons theory.

\subsubsection{The underlying raw $L_\infty$-algebra $\mathcal E(X)$}
Let $\g$ be a Lie algebra with Lie bracket $[\cdot, \cdot] : \g \otimes \g \to \g$ 
and equipped with a non-degenerate invariant bilinear pairing 
$\langle \cdot, \cdot \rangle : \g \otimes \g \to \bbC$. Given a choice of 
meromorphic $1$-form $\omega$ on $C$, the action of $4$-dimensional 
semi-holomorphic Chern-Simons theory on $X$ takes the form
\begin{equation} \label{4d CS action}
S_{\rm 4dCS}[A] \;=\; \frac{\ii}{2\pi}\int_X \omega \wedge \big\langle A \, \overset{\wedge}, \, \tfrac 12 (\d_\Sigma + \bar\partial) A + \tfrac{1}{3!} [A\, \overset{\wedge}, \, A] \big\rangle
\end{equation}
for a $\g$-valued $1$-form field $A$ on $X$. By analogy with the Chern-Simons 
theory setup of Example \ref{ex: 3d CS}, the mixed topological-holomorphic 
nature of the action \eqref{4d CS action} along the two factors $\Sigma$ and $C$ 
suggests working with the cochain complex of $\g$-valued mixed de Rham-Dolbeault forms
\begin{equation} \label{Tot dR-Dol}
\big( \mathcal E(X), \d \big) \;\coloneqq\; \Big( \textup{Tot}^{\oplus} \big( \g \otimes \Omega^{\bullet, (0, \bullet)}(X) \big), \d_\Sigma + \bar\partial \Big)
\end{equation}
defined as the totalization of the bicomplex, see \cite[(3.2)]{BSV2},
\begin{equation} \label{dR-Dol bicpx}
\big( \g \otimes \Omega^{\bullet, (0, \bullet)}(X), \d_\Sigma, \bar\partial \big) \;\coloneqq\; \g \otimes \big( \Omega^{\bullet}(\Sigma), \d_\Sigma \big) \; \widehat{\otimes}\; \big( \Omega^{(0,\bullet)}(C), \bar\partial \big) ,
\end{equation}
where $\widehat{\otimes}$ is the completed projective tensor product of 
locally convex topological vector spaces. The cochain complex \eqref{Tot dR-Dol} 
can be written out explicitly as follows
\begin{equation}
\begin{tikzcd}[
  column sep=15mm,
  row sep=6mm
]
\g \otimes \Omega^{0,(0,0)}(X) \arrow[r, "\d_\Sigma + \bar\partial"] & |[alias=A]| \g \otimes \Omega^{1,(0,0)}(X) \oplus \g \otimes \Omega^{0,(0,1)}(X) & \\
& |[alias=B]| \g \otimes \Omega^{1,(0,1)}(X) \oplus \g \otimes \Omega^{2,(0,0)}(X) \arrow[r, "\d_\Sigma + \bar\partial"] & \g \otimes \Omega^{2,(0,1)}(X) .
\arrow["\d_\Sigma + \bar\partial"' near end, from=A, to=B, out=0, in=180, looseness=1.8, overlay]
\end{tikzcd}
\end{equation}
This is canonically a differential graded Lie algebra, i.e.\ an $L_\infty$-algebra 
with $\ell_n = 0$ for $n \geq 3$, where the $\ell_2$-bracket is given by 
$\ell_2 : \mathcal E(X)^{\otimes 2} \to \mathcal E(X)$, $\alpha \otimes \beta 
\mapsto [\alpha \,\overset{\wedge},\,  \beta]$. However, the $L_\infty$-algebra 
$\big( \mathcal E(X), (\d, \ell_2) \big)$ does not yet describe $4$-dimensional 
semi-holomorphic Chern-Simons theory as it does not come equipped with a suitable cyclic structure.
\sk

Indeed, crucially, the field $A$ in \eqref{4d CS action} is allowed to develop 
singularities at the zeros of $\omega$ and in order for the action functional 
to be well-defined and gauge invariant, the field $A$ should also satisfy certain 
boundary conditions at the poles of $\omega$, see for instance \cite{CY3, DLMV, BSV}. 
This motivates twisting the ordinary Dolbeault complexes on $C$ entering 
\eqref{dR-Dol bicpx} by suitable holomorphic line bundles in order to accommodate 
these singularities and boundary conditions.

\subsubsection{Twisted Dolbeault complexes}
Denote by $\mathcal O$, $\mathcal O^\ast$, $\mathcal M$ and $\mathcal M^{(1)}$ 
the sheaves of holomorphic functions, non-vanishing holomorphic functions, meromorphic 
functions and meromorphic $1$-forms on $C = \CP$, respectively.
\sk

A divisor $D$ on $C$ is a function $D : C \to \bbZ$ whose support 
$\supp(D) \coloneqq \{ p \in C \,|\, D(p) \neq 0 \}$ is a finite subset 
of $C$. The degree of $D$ is given by the finite sum $\deg(D) \coloneqq \sum_{p \in C} D(p)$. 
Given two divisors $D, D' : C \to \bbZ$ we write $D \leq D'$ if $D(p) \leq D'(p)$ for all $p \in C$.
\sk

To every $f \in \mathcal M(C) \setminus \{0\}$ is associated its principal divisor 
$(f) : C \to \bbZ$, encoding the locations and orders of its poles and zeros. In particular, we have 
a canonical decomposition $(f) = (f)_0 + (f)_\infty$ into the divisors 
$(f)_0, (f)_\infty : C \to \bbZ$ of zeros and poles of $f$, with $(f)_0 \geq 0$ and 
$(f)_\infty \leq 0$. We always have $\deg (f) = 0$. Likewise, every 
$\omega \in \mathcal M^{(1)}(C) \setminus \{ 0 \}$ defines a divisor 
$(\omega) : C \to \bbZ$ with its canonical decomposition
\begin{equation}
(\omega) = (\omega)_0 + (\omega)_\infty
\end{equation}
into the divisors of zeros and poles $(\omega)_0, (\omega)_\infty : C \to \bbZ$ 
with $(\omega)_0 \geq 0$ and $(\omega)_\infty \leq 0$. Since we are working on 
$C = \CP$, we have $\deg (\omega) = -2$.
\sk

We fix a choice of meromorphic $1$-form $\omega \in \mathcal M^{(1)}(C)\setminus \{0\}$, 
which will correspond to the one entering the action \eqref{4d CS action}. 
In what follows, we will only be interested in divisors $D:C\to\mathbb Z$ 
which are built out of $(\omega)$ in the sense that $D = D_0 + D_\infty$ 
with $0 \leq D_0 \leq (\omega)_0$ and $(\omega)_\infty \leq D_\infty \leq 0$.
\sk

To any divisor $D : C \to \bbZ$ one associates a holomorphic line bundle 
$L_D \to C$ and to any divisor inequality $D \leq D'$ one associates a holomorphic line 
bundle morphism $L_D \to L_{D'}$, as recalled in \cite[Construction 2.3]{BSV2}. 
In particular, the holomorphic line bundle $L_D \to C$ is constructed explicitly as follows:
Let $p_0, p_\infty \in C \setminus \supp (\omega)$ with $p_0 \neq p_\infty$, and 
introduce the open subsets $U_0 \coloneqq C \setminus \{p_\infty \}$ and 
$U_\infty \coloneqq C \setminus \{ p_0 \}$. Let $z:U_0\to\bbC$ 
and $w : U_\infty\to \bbC$ be holomorphic coordinates with $z(p_0)=0$
and $w(p_\infty)=0$ such that $w = 1/z$ on the overlap $U_{0\infty} \coloneqq U_0\cap U_\infty$.
We choose a pair of meromorphic functions 
$\psi^D_0 \in \mathcal M(U_0)$ and $\psi^D_\infty \in \mathcal M(U_\infty)$ 
such that $(\psi^D_0) = (\psi^D_\infty) = D$ by
\begin{equation} \label{psi 0 and infty}
\psi^D_0(p) \coloneqq \prod_{p' \in \supp(D)} \big( z(p) - z(p') \big)^{D(p')}, \qquad
\psi^D_\infty(q) \coloneqq \prod_{p' \in \supp(D)} \big( 1 - w(p')^{-1} w(q) \big)^{D(p')},
\end{equation}
for $p \in U_0$ and $q \in U_\infty$. On the overlap $U_{0\infty} = U_0\cap U_\infty$, 
their ratio is holomorphic and nowhere vanishing, and defines the transition function 
$g_{0\infty}^D \in \mathcal O^\ast(U_{0\infty})$ of the holomorphic line bundle $L_D \to C$, i.e.\
\begin{equation} \label{transition function}
g_{0\infty}^D(p) \coloneqq \frac{\psi^D_0(p)}{\psi^D_\infty(p)} = z(p)^{\deg(D)},
\end{equation}
for any $p \in U_{0\infty}$, where in the last step we have used the relation $w(p)^{-1} = z(p)$. 
\sk

A smooth $L_D$-valued $(0,r)$-form $s \in \Omega^{0,r}(C, L_D)$, for 
$r \in \{0,1\}$, consists of a pair $s = (s_0, s_\infty)$ of smooth $(0,r)$-forms 
$s_0 \in \Omega^{0,r}(U_0)$ and $s_\infty \in \Omega^{0,r}(U_\infty)$ satisfying the gluing condition
\begin{equation}
s_0(p) = g_{0\infty}^D(p) s_\infty(p),
\end{equation}
for $p \in U_{0\infty}$. The twisted Dolbeault complex associated with $D$ is then
\begin{equation} \label{twisted Dolbeault complex}
\big( \Omega^{0,\bullet}(C,L_D),\bar\partial \big),
\end{equation}
where $\bar\partial$ acts component-wise in the above local description. 
This action is well-defined since the transition function \eqref{transition function} is holomorphic.
\sk

One can equivalently model the complex \eqref{twisted Dolbeault complex} as follows:
Let $\Gamma_D^{0,r}(C)$, for $r \in \{0,1\}$, denote the space of smooth $(0,r)$-forms 
$\sigma \in \Omega^{0,r}\big( C\setminus\supp(D) \big)$ with prescribed local behavior 
at every $q\in\supp(D)$ given in any local coordinate $t$ centred at $q$ by
\begin{equation}
\sigma = t^{-D(q)} \widetilde\sigma,
\end{equation}
for some smooth $(0,r)$-form $\widetilde\sigma$ defined in a neighborhood of $q$. 
In other words, positive values of $D(q)$ correspond to poles of order at most $D(q)$, 
while negative values correspond to zeros of order at least $-D(q)$. We then have an isomorphism
\begin{equation} \label{Gamma O iso}
\Gamma_D^{0,r}(C) \;\overset{\cong}\longrightarrow\; \Omega^{0,r}(C,L_D), \quad
\sigma \;\longmapsto \; \big( \psi_0^D\,\sigma|_{U_0}, \psi_\infty^D\,\sigma|_{U_\infty} \big),
\end{equation}
under which the Dolbeault differential on $\Omega^{0,\bullet}(C,L_D)$ is 
represented by the ordinary Dolbeault differential on $\Omega^{0,\bullet}\big( C\setminus\supp(D) \big)$.

\subsubsection{Allowing singular structures: the $L_\infty$-algebra $\mathcal L(X)$} \label{sec:L(X) def}
We allow for prescribed singularities at the zeros of $\omega$ following the approach of 
\cite[Section 3.2]{BSV2}. To assign different singular behaviors to the light-cone 
components of the $4$-dimensional semi-holomorphic Chern-Simons field $A$, we make 
use of the Hodge star operator $\ast_\Sigma$ to decompose
\begin{equation}
\Omega^1(\Sigma)=\Omega^+(\Sigma)\oplus \Omega^-(\Sigma),
\qquad
\ast_\Sigma \alpha^\pm=\pm \alpha^\pm
\end{equation}
into $\pm$-eigenspaces associated with a choice of conformal class 
of Lorentzian metrics on $\Sigma$ and choose a decomposition of the zero divisor
\begin{equation} \label{zero divisor splitting}
(\omega)_0=(\omega)^+_0+(\omega)^-_0 ,
\qquad
(\omega)^\pm_0\geq 0 .
\end{equation}
The de Rham-Dolbeault bicomplex $\Omega^{\bullet, (0, \bullet)}(X)$ 
underlying $\mathcal E(X)$ in \eqref{Tot dR-Dol} is then replaced by the divisor-twisted bicomplex
\begin{equation} \label{bicomplex sing}
\Omega^{\bullet,(0,\bullet)}_{\rm sgl}(X)
\;\coloneqq\;
\left(
\begin{tikzcd}[column sep=10mm]
\Omega^{0,(0,\bullet)}(X)
\arrow[r, "\d_\Sigma"] &
{\displaystyle \bigoplus_{i=\pm}} \; \Omega^{i,(0,\bullet)}\big( X,L_{(\omega)^i_0} \big) \arrow[r, "\d_\Sigma"] &
\Omega^{2,(0,\bullet)}\big( X, L_{(\omega)_0} \big)
\end{tikzcd}
\right).
\end{equation}
The vertical Dolbeault differential $\bar\partial$ along $C$ is 
suppressed here and the horizontal differentials $\d_\Sigma = \d_\Sigma^+ + \d_\Sigma^-$ 
are given by the tensor product of the light-cone components of the de Rham 
differentials along $\Sigma$ with corresponding holomorphic line bundle morphisms induced by 
the divisor inequalities $0\leq (\omega)^\pm_0\leq (\omega)_0$, see 
\cite[Construction 2.3]{BSV2}. Explicit formulae for $\d_\Sigma^\pm$ in the 
context of the application considered in the present paper are given in 
\eqref{d pm Lax} below.
\sk

Tensoring the bicomplex \eqref{bicomplex sing} with $\g$ and forming its totalization 
yields the cochain complex
\begin{equation} \label{L X def}
\big(\mathcal L(X),\ell_1\big)
\;\coloneqq\;
\left(
\textup{Tot}^\oplus\big(\g\otimes\Omega^{\bullet,(0,\bullet)}_{\rm sgl}(X)\big),
\d_\Sigma+\bar\partial
\right),
\end{equation}
which can be written out explicitly as
\begin{equation}
\begin{tikzcd}[
  column sep=5mm,
  row sep=6mm
]
\g \otimes \Omega^{0,(0,0)}(X) \arrow[r, "\d_\Sigma + \bar\partial"] & |[alias=A]| {\displaystyle \bigoplus_{i=\pm}} \; \g \otimes \Omega^{i,(0,0)}\big( X, L_{(\omega)_0^i} \big) \oplus \g \otimes \Omega^{0,(0,1)}(X) & \\
& |[alias=B]| {\displaystyle \bigoplus_{i=\pm}} \; \g \otimes \Omega^{i,(0,1)}\big( X, L_{(\omega)_0^i} \big) \oplus \g \otimes \Omega^{2,(0,0)}\big( X, L_{(\omega)_0} \big) \arrow[r, "\d_\Sigma + \bar\partial"] & \g \otimes \Omega^{2,(0,1)}\big( X, L_{(\omega)_0} \big).
\arrow["\d_\Sigma + \bar\partial"' near end, from=A, to=B, out=0, in=180, looseness=1.6, overlay]
\end{tikzcd}
\end{equation}
Importantly, the assignment of holomorphic line bundles in \eqref{bicomplex sing} 
ensures that the $\wedge$-product is well defined using the canonical isomorphism
\begin{equation}
L_{(\omega)^+_0}\otimes L_{(\omega)^-_0} \cong L_{(\omega)_0},
\end{equation}
which follows from the decomposition \eqref{zero divisor splitting}.
As a result, the totalization \eqref{L X def} comes equipped with the 
structure of a differential graded Lie algebra with Lie bracket
\begin{equation} \label{Lie bracket L(X)}
\ell_2 \;:\; \mathcal L(X)^{\otimes 2} \longrightarrow \mathcal L(X) , \qquad (\alpha,\beta) \longmapsto [\alpha \,\overset{\wedge}{,}\, \beta].
\end{equation}
We write $\ell \coloneqq (\ell_1, \ell_2)$ and denote the corresponding 
differential graded Lie algebra as $\big( \mathcal L(X), \ell \big)$.

\subsubsection{Imposing boundary conditions: the $L_\infty$-algebra $\mathcal F(X)$} \label{sec:F(X) def}
To extract a $2$-dimensional integrable field theory from $4$-dimensional semi-holomorphic 
Chern-Simons theory one must also impose boundary conditions at the poles of $\omega$. 
According to \cite[Section 3.3]{BSV2}, this amounts to refining the previous twisting 
of the Dolbeault complexes on $C$ appearing in the bicomplex \eqref{bicomplex sing} 
by non-positive divisors supported at the poles of $\omega$.
\sk

A choice of local boundary condition consists of a pair of decompositions
\begin{subequations} \label{pole divisor split}
\begin{equation} \label{saturation conditions}
(\omega)^0_\infty+(\omega)^2_\infty=(\omega)_\infty = (\omega)^+_\infty+(\omega)^-_\infty,
\end{equation}
of the pole divisor of $\omega$, satisfying the inequalities
\begin{equation} \label{pole divisor ineq}
(\omega)^0_\infty\leq (\omega)^\pm_\infty\leq (\omega)^2_\infty\leq 0
\end{equation}
and such that
\begin{equation} \label{degree conditions}
\deg\big((\omega)^+_0+(\omega)^+_\infty\big) = \deg\big((\omega)^-_0+(\omega)^-_\infty\big) = g - 1 = -1 ,
\end{equation}
\end{subequations}
where we took into account that $C=\CP$ has genus $g=0$.
The divisor-twisted de Rham-Dolbeault bicomplex $\Omega^{\bullet, (0, \bullet)}_{\rm sgl}(X)$ 
from \eqref{bicomplex sing} is then further replaced by the bicomplex with both 
singularities and boundary conditions
\begin{equation} \label{sgl bdy bicomplex}
\Omega^{\bullet,(0,\bullet)}_{\rm sgl,bdy}(X)
\coloneqq
\left(\!\!
\begin{tikzcd}[column sep=4mm]
\Omega^{0,(0,\bullet)}\big( X,L_{(\omega)^0_\infty} \big)
  \arrow[r, "\d_\Sigma"]
&
{\displaystyle \bigoplus_{i=\pm}} \; \Omega^{i,(0,\bullet)}\big( X,L_{(\omega)^i_0 + (\omega)^i_\infty} \big)
  \arrow[r, "\d_\Sigma"]
&
\Omega^{2,(0,\bullet)}\big( X,L_{(\omega)_0+(\omega)^2_\infty} \big)
\end{tikzcd}
\!\!\right).
\end{equation}
Just as in \eqref{bicomplex sing}, the vertical Dolbeault differentials 
$\bar\partial$ along $C$ are suppressed and the definition of the horizontal 
differentials $\d_\Sigma = \d_\Sigma^+ + \d_\Sigma^-$ makes use of the holomorphic line bundle 
morphisms induced by the divisor inequalities \eqref{pole divisor ineq}, see 
\cite[Construction 2.3]{BSV2}. In the context of the application considered in the present paper, 
explicit formulae for the light-cone components of the horizontal differential are displayed in \eqref{d pm}.
\sk

The totalization of the bicomplex \eqref{sgl bdy bicomplex} tensored with $\g$, namely
\begin{equation} \label{F X def}
\big(\mathcal F(X),\ell_1\big)
\;\coloneqq\;
\left(
\textup{Tot}^\oplus\big(\g\otimes\Omega^{\bullet,(0,\bullet)}_{\rm sgl,bdy}(X)\big),
\d_\Sigma+\bar\partial
\right),
\end{equation}
admits a Lie bracket
\begin{equation}
\ell_2 \;:\; \mathcal F(X)^{\otimes 2} \longrightarrow \mathcal F(X) , \qquad (\alpha,\beta) \longmapsto [\alpha \,\overset{\wedge}{,}\, \beta]
\end{equation}
similar to \eqref{Lie bracket L(X)} with the additional subtlety 
that the $\wedge$-product on the bicomplex \eqref{sgl bdy bicomplex} now 
involves non-trivial holomorphic line bundle morphisms $L_D \to L_{D'}$ for some 
divisor inequalities $D \leq D'$ determined by \eqref{pole divisor split}, 
see for instance \eqref{eqn:Lie} below in the particular example to be considered in this paper.
We write $\ell \coloneqq (\ell_1, \ell_2)$ and denote the resulting differential 
graded Lie algebra by $\big( \mathcal F(X), \ell \big)$.
Importantly, this $L_\infty$-algebra admits a cyclic structure on compactly supported sections
\begin{equation} \label{cyclic structure FX}
\langle\!\langle \cdot, \cdot \rangle\!\rangle_\omega \;:\; \mathcal F_c(X) \otimes \mathcal F_c(X) \longrightarrow \bbC , \qquad
(\alpha, \beta) \longmapsto \frac{\ii}{2 \pi} \int_X \omega \wedge \langle \alpha \,\overset{\wedge}{,}\, \beta \rangle
\end{equation}
and the corresponding Maurer-Cartan action \eqref{MC action} coincides 
exactly with the action of $4$-dimensional semi-holomorphic Chern-Simons theory in \eqref{4d CS action}.

\subsection{$2$-dimensional integrable field theories from homotopy transfer}
A fundamental feature of $4$-dimensional semi-holomorphic Chern-Simons theory 
on $X = \Sigma \times C$ is that, after ``integrating out $C=\CP$'', it gives rise to 
$2$-dimensional integrable field theories on the Lorentzian spacetime $\Sigma$ \cite{CY3}. 
From the general perspective recalled in Section \ref{sec:pertFT and Linfty}, such a 
field theory on $\Sigma$ should be (perturbatively) encoded by a cyclic $L_\infty$-algebra $(\F(\Sigma),\ell')$.
\sk

In the approach of \cite{BSV2}, this cyclic $L_\infty$-algebra is obtained from the 
cyclic $L_\infty$-algebra $(\mathcal F(X),\ell)$ of $4$-dimensional semi-holomorphic 
Chern-Simons theory with singularities and boundary conditions, reviewed in Section 
\ref{sec:F(X) def}, by homotopy transfer to its partial cohomology along $C$. This gives 
a precise homological meaning to the operation of ``integrating out $C$''. Moreover, the 
construction of the Lax connection of the resulting $2$-dimensional integrable field 
theory uses an analogous homotopy transfer, now applied to the $L_\infty$-algebra 
$(\mathcal L(X),\ell)$ of $4$-dimensional semi-holomorphic Chern-Simons theory with 
singularities only, reviewed in  Section \ref{sec:L(X) def}. Before recalling these 
constructions in detail, we begin with a brief review of the relevant aspects of 
the homotopy transfer theorem, specializing to the case needed in the present setting.

\subsubsection{\label{sec:HTT}Homotopy transfer theorem}
A \emph{strong deformation retract} $(i,p,h)$ from a cochain complex $(L,\d)$ 
to another cochain complex $(L', \d')$ is a diagram
\begin{subequations} \label{SDR def}
\begin{equation} \label{SDR diagram}
\begin{tikzcd}
(L', \d') \arrow[r, "i", shift left=1mm] & (L, \d) \arrow[l, "p", shift left=1mm] \arrow["h"', out=-15,in=15,distance=6mm]
\end{tikzcd}
\end{equation}
consisting of a pair of cochain maps $i : L' \to L$ and $p : L \to L'$, i.e.\ 
degree $0$ linear maps such that $\d \, i = i \, \d'$ and $\d' \, p = p \, \d$, 
and a degree $-1$ linear map $h : L \to L$, satisfying the identities
\begin{equation} \label{SDR relations}
p \, i \;=\; \id_{L'}, \qquad
i\,p \;=\; \id_L + \d \, h + h \, \d,
\end{equation}
together with the side conditions
\begin{equation} \label{SDR side conditions}
h \, i = 0, \qquad
p \, h = 0, \qquad
h^2 = 0.
\end{equation}
\end{subequations}
It follows from \eqref{SDR relations} that $i$ and $p$ are both 
quasi-isomorphisms of cochain complexes which are quasi-inverses of each other.
\sk

An important feature of $L_\infty$-algebra structures is that they can be 
transferred along strong deformation retracts. The formal statement is the content of 
the \emph{homotopy transfer theorem} for $L_\infty$-algebras. For a proof in the general 
context of operad algebras we refer the reader to \cite[Chapter 10.3]{LodayVallette}.
\begin{theo} \label{thm: Homotopy transfer}
Let $(L,\d)$ be the underlying cochain complex of an $L_\infty$-algebra $(L, \ell)$, 
i.e.\ with $\ell_n$-brackets $\ell = (\ell_n)_{n \geq 1}$ such that $\ell_1 = \d$. 
Given a strong deformation retract to another cochain complex $(L', \d')$ as in \eqref{SDR def}, 
there exists a transferred $L_\infty$-algebra structure $\ell' = (\ell'_n)_{n \geq 1}$ on 
$L'$ with $\ell'_1 = \d'$. Moreover, the $L_\infty$-algebras $(L, \ell)$ and $(L', \ell')$ 
are weakly equivalent in the sense that the quasi-isomorphism $i$ extends to an $\infty$-quasi-isomorphism
\begin{equation} \label{i infty HTT}
\begin{tikzcd}
i_\infty \; : \; (L', \ell') \arrow[r, smooth squiggle, "\sim"] & (L, \ell),
\end{tikzcd}
\end{equation}
which consists of a collection of degree $1-n$ linear maps $i_n : L'^{\otimes n} \to L$ 
for $n \in \bbZ_{\geq 1}$, with $i_1 = i$, satisfying the $L_\infty$-morphism identities 
relating the transferred $L_\infty$-algebra structure $\ell'$ on $L'$ to the original 
$L_\infty$-algebra structure $\ell$ on $L$.
\end{theo}

For the main computation to be carried out in Section \ref{sec:application} we shall 
need the explicit formulae for the transferred $\ell'_n$ brackets and the components 
$i_n$ of the $\infty$-quasi-isomorphism \eqref{i infty HTT}. In fact, in the context of 
\cite{BSV2}, the homotopy transfer theorem is applied only in the special case when the 
$L_\infty$-algebra $(L,\ell)$ is a differential graded Lie algebra, specifically 
$(\mathcal L(X), \ell)$ from Section \ref{sec:L(X) def} or $(\mathcal F(X), \ell)$ 
from Section \ref{sec:F(X) def}. We will therefore only recall the relevant formulae 
in this simpler setting. For the general formulae, see for instance \cite{JalaliFarahani:2023sfq}.
\sk

Suppose that $(L, \ell)$ is a differential graded Lie algebra with $\ell = (\d, \ell_2)$. 
Then the components $i_n : L'^{\otimes n} \to L$ for $n \in \bbZ_{\geq 1}$ of the 
$\infty$-quasi-isomorphism \eqref{i infty HTT} from the homotopy transfer theorem are 
given by $i_1(a) \coloneqq i(a)$ for $a \in L'$ and then recursively for all $n \geq 2$ by
\begin{multline} \label{in def}
i_n(a_1, \ldots, a_n) \coloneqq \frac{1}{2} \sum_{m=1}^{n-1} \sum_{\sigma \in \text{Sh}(m, n-m)} \chi(\sigma; a_1, \ldots, a_n)\, \zeta(\sigma; a_1, \ldots, a_n)\\
h \Big( \ell_2\big( i_m(a_{\sigma(1)}, \ldots, a_{\sigma(m)}), i_{n-m}(a_{\sigma(m + 1)}, \ldots, a_{\sigma(n)}) \big) \Big),
\end{multline}
for all homogeneous $a_1, \ldots, a_n \in L'$, where $\chi(\sigma; a_1, \ldots, a_n)$ is the Koszul sign defined by
\begin{equation} \label{chi sign def}
a_1 \wedge \ldots \wedge a_n = \chi(\sigma; a_1, \ldots, a_n)\, a_{\sigma(1)} \wedge \ldots \wedge a_{\sigma(n)}
\end{equation}
and $\zeta(\sigma; a_1, \ldots, a_n)$ is the sign given in the present setting by
\begin{equation} \label{zeta sign def}
\zeta(\sigma; a_1, \ldots, a_n) \coloneqq (-1)^{m (n-m) + m + (1-n+m) \sum_{k=1}^m |a_{\sigma(k)}|}.
\end{equation}
The transferred $\ell'_n$-brackets $\ell'_n : L'^{\otimes n} \to L'$ are then defined for all $n \geq 2$ by
\begin{multline} \label{elln def}
\ell'_n(a_1, \ldots, a_n) \coloneqq \frac{1}{2} \sum_{m=1}^{n-1} \sum_{\sigma \in \text{Sh}(m, n-m)} \chi(\sigma; a_1, \ldots, a_n)\, \zeta(\sigma; a_1, \ldots, a_n)\\
p \Big( \ell_2\big( i_m(a_{\sigma(1)}, \ldots, a_{\sigma(m)}), i_{n-m}(a_{\sigma(m + 1)}, \ldots, a_{\sigma(n)}) \big) \Big),
\end{multline}
for all homogeneous $a_1, \ldots, a_n \in L'$.

\subsubsection{The $L_\infty$-algebra $\F(\Sigma)$} \label{sec: FSigma def}
It was proved in \cite[Propositions 3.13 and 3.16]{BSV2} that the cyclic $L_\infty$-algebra 
$\big( \mathcal F(X), \ell \big)$ of $4$-dimensional semi-holomorphic Chern-Simons theory 
with singularities and boundary conditions, reviewed in Section \ref{sec:F(X) def}, admits 
a weakly equivalent description by the cyclic $L_\infty$-algebra $\big( \mathcal F(\Sigma), \ell' \big)$, where 
the spectral parameter $C = \CP$ has been ``integrated out'', with $\infty$-quasi-isomorphism 
\begin{equation} \label{infty-quasi-iso i}
\begin{tikzcd}
\widetilde{i}_\infty \; : \; \big( \mathcal F(\Sigma), \ell' \big) \arrow[r, smooth squiggle, "\sim"] & \big( \mathcal F(X), \ell \big).
\end{tikzcd}
\end{equation}
To set the stage for the main computations in the paper, it is useful to briefly 
recall from \cite{BSV2} the strategy towards the construction of the strong deformation 
retract $(\widetilde{i},\widetilde{p},\widetilde{h})$ that ``integrates out $C$'', thus 
leading to the cyclic $L_\infty$-algebra $\big( \mathcal F(\Sigma), \ell' \big)$ and 
the $\infty$-quasi-isomorphism \eqref{infty-quasi-iso i} via homotopy transfer, as recalled 
in Section \ref{sec:HTT}. Full details for the specific class of models considered in this 
paper will be given in Section \ref{SDR for FX} below.
\sk

The construction of $(\widetilde{i},\widetilde{p},\widetilde{h})$ makes essential use 
of strong deformation retracts for divisor-twisted Dolbeault complexes to their cohomologies,
\begin{equation} \label{SDRs all D prelim}
\begin{tikzcd}
\mathsf H^\bullet \Omega^{0,\bullet}(C, L_D) \arrow[r, "i_D", shift left=1mm] & \Omega^{0,\bullet}(C, L_D) \arrow[l, "p_D", shift left=1mm] \arrow["h_D"', out=-15,in=15,distance=8mm]
\end{tikzcd}
\end{equation}
for any divisor $D : C \to \bbZ$, which exist by Hodge theory, see 
\cite[Appendix A.2]{BSV2}. Given a choice of such strong deformation retracts, 
the homological perturbation lemma was used to build the strong deformation retract
$(\widetilde{i},\widetilde{p},\widetilde{h})$ relating the cochain complex underlying 
$\big( \mathcal F(X), \ell\big)$ and its partial cohomology along $C$, denoted 
$\mathcal{F}(\Sigma)$. The homotopy transfer theorem could then be applied in 
\cite[Proposition 3.13]{BSV2} to implicitly define the $L_\infty$-algebra 
$\big( \mathcal F(\Sigma), \ell' \big)$ together with the corresponding 
$\infty$-quasi-isomorphism \eqref{infty-quasi-iso i}. It was further shown in 
\cite[Proposition 3.16]{BSV2} that the family of strong deformation retracts 
\eqref{SDRs all D prelim} can be chosen to be compatible with the cyclic 
structure on $\big( \mathcal F(X), \ell \big)$, ensuring that the latter 
can be transferred to $\big( \mathcal F(\Sigma), \ell' \big)$. 
\sk

In Section \ref{SDR for FX} we will make the above strategy concrete 
starting from explicit strong deformation retracts as in \eqref{SDRs all D prelim} 
that will be constructed in Section \ref{sec: maps i p h deg pos}.

\subsubsection{\label{sec:L(Sigma)}Lax connections and the $L_\infty$-algebra $\L(\Sigma)$}
It was proved in \cite[Proposition 3.6]{BSV2} that the $L_\infty$-algebra 
$\big( \mathcal L(X), \ell \big)$ from Section \ref{sec:L(X) def} also admits 
a weakly equivalent description by the $L_\infty$-algebra $\big( \mathcal L(\Sigma), \ell' \big)$, 
obtained along the same lines as in Section \ref{sec: FSigma def} by homotopy 
transfer to the partial cohomology along $C$. In contrast to the field-theoretic 
$L_\infty$-algebra $\big( \mathcal F(\Sigma), \ell' \big)$, the transferred 
$L_\infty$-algebra $\big( \mathcal L(\Sigma), \ell' \big)$ has vanishing higher 
brackets $\ell^\prime_n = 0$ for $n \geq 3$. In particular, both $\big( \mathcal L(X), \ell \big)$ 
and $\big( \mathcal L(\Sigma), \ell' \big)$ can be regarded as a differential graded 
Lie algebra. Furthermore, it turns out that the strong deformation retract 
$(i, \widetilde{p},h)$ relating the cochain complexes underlying 
$\big( \mathcal L(X), \ell \big)$ and $\big( \mathcal L(\Sigma), \ell' \big)$ 
is particularly well-behaved in that $i$ preserves the Lie brackets and 
hence provides a quasi-isomorphism of differential graded Lie algebras
\begin{equation} \label{i infty-quasi-iso prelim}
\begin{tikzcd}
i \; : \; \big( \mathcal L(\Sigma), \ell' \big) \arrow[r, "\sim"] & \big( \mathcal L(X), \ell \big).
\end{tikzcd}
\end{equation}
On the other hand, $\widetilde{p}$ fails to preserve the Lie brackets and 
thereby extends to an $\infty$-quasi-isomorphism 
\begin{equation} \label{infty-quasi-iso p}
\begin{tikzcd}
\widetilde{p}_\infty \; : \; \big( \mathcal L(X), \ell \big) \arrow[r, smooth squiggle, "\sim"] & \big( \mathcal L(\Sigma), \ell' \big),
\end{tikzcd}
\end{equation}
which provides a quasi-inverse to \eqref{i infty-quasi-iso prelim}.
\sk

As argued in \cite{BSV2}, the transferred $L_\infty$-algebra 
$\big( \mathcal L(\Sigma), \ell' \big)$ should be thought of as the space of 
Lax connections for integrable field theories on $\Sigma$ associated with the 
meromorphic $1$-form $\omega$ and the choice of splitting of its zero divisor as 
in \eqref{zero divisor splitting}. Indeed, a Maurer-Cartan element in the 
$L_\infty$-algebra $\big( \mathcal L(\Sigma), \ell' \big)$ is given by a meromorphic 
$\g$-valued $1$-form $L = L_+ + L_-$ along $\Sigma$, with light-cone components 
$L_\pm\in \Omega^\pm(\Sigma,\g)\otimes \mathcal O_{(\omega)^\pm_0}(C)$, satisfying the flatness equation
\begin{equation}
\d_\Sigma L+\tfrac12 \; [L\,\overset{\wedge}{,}\, L]=0 .
\end{equation}
Furthermore, the Lax connection of the $2$-dimensional integrable field theory 
associated with the cyclic $L_\infty$-algebra $\big( \F(\Sigma), \ell' \big)$ 
was given a succinct homological interpretation in \cite[Section 3.4]{BSV2} in 
terms of the $L_\infty$-morphism $\lambda_\infty = (\lambda_n)_{n\geq 1}$ defined 
through the following composition
\begin{equation} \label{lambda infty def}
\begin{tikzcd}[column sep=12mm, row sep=10mm]
(\mathcal F(\Sigma), \ell') \arrow[d, smooth squiggle, "\sim", "\widetilde{i}_\infty"'] \arrow[r, smooth squiggle, "\lambda_\infty"] & (\mathcal L(\Sigma), \ell')\\
(\mathcal F(X), \ell) \arrow[r, "\psi"'] & (\mathcal L(X), \ell) \arrow[u, smooth squiggle, "\sim"', "\widetilde{p}_\infty"]
\end{tikzcd}
\end{equation}
where the left and right vertical $\infty$-quasi-isomorphisms are given in 
\eqref{infty-quasi-iso i} and \eqref{infty-quasi-iso p}, respectively. The bottom 
strict $L_\infty$-algebra morphism is induced from the canonical bundle morphisms 
$L_D \to L_{D'}$ for $D \leq D'$ which drops the negative divisors in $\F(X)$ 
that are not present in $\L(X)$, see Sections \ref{sec:L(X) def} and \ref{sec:F(X) def}.


\section{\label{sec:explicitSDR}Explicit strong deformation retracts}
In this section we give explicit models for the strong deformation retracts from 
the divisor-twisted Dolbeault complexes on $C=\CP$ to their cohomologies, as in 
\eqref{SDRs all D prelim}, for all divisors $D : C \to \bbZ$. The cases $\deg(D) \geq 0$, 
$\deg(D) = -1$ and $\deg(D) \leq -2$ are treated separately. The verification of the 
strong deformation retract identities and side conditions, which relies on standard 
results from complex analysis, is deferred to Appendix \ref{sec:proof of SDR}.

\subsection{Case \texorpdfstring{$\deg(D) \geq 0$}{deg D geq 0}} \label{sec: maps i p h deg pos}
The cohomology of $\Omega^{0, \bullet}(C, L_D)$ for $\deg(D) \geq 0$ is concentrated 
in degree $0$ and is isomorphic to $\bbC^{1+\deg(D)}$.
Define the cochain map
\begin{subequations}\label{i def 0}
\begin{equation}
i_D : \bbC^{1+\deg(D)} \longrightarrow \Omega^{0, \bullet}(C, L_D)
\end{equation}
sending $\alpha = (\alpha_j)_{j=0}^{\deg(D)} \in \bbC^{1+\deg(D)}$ to 
$i_D(\alpha) = (i_D(\alpha)_0, i_D(\alpha)_\infty) \in \Omega^{0, 0}(C, L_D)$, where
\begin{align}
i_D(\alpha)_0(z) &:= \sum_{j=
0}^{\deg(D)} \alpha_j z^j, \\
i_D(\alpha)_\infty(w) &
:= \sum_{j=0}^{\deg(D)} \alpha_j w^{\deg(D)-j}
\end{align}
\end{subequations}
are defined with respect to the chosen local charts.
Define the cochain map
\begin{subequations} \label{p def 0}
\begin{equation}
p_D : \Omega^{0, \bullet}(C, L_D) \longrightarrow \bbC^{1+\deg(D)},
\end{equation}
to act trivially on $\Omega^{0, 1}(C, L_D)$ and by sending $s = 
(s_0, s_\infty) \in \Omega^{0, 0}(C, L_D)$ to the coefficients 
$\big( p_D(s)_j \big)_{j=0}^{\deg(D)}$ of the polynomial
\begin{equation} \label{p def 0 b}
\sum_{j=0}^{\deg(D)} p_D(s)_j z^j \coloneqq - \frac{1 + \deg(D)}{2 \pi \ii} \int_\bbC \frac{s_0(z') (1+ z \bar z')^{\deg(D)}}{(1+|z'|^2)^{2 + \deg(D)}} \, \d z' \wedge \d \bar z'.
\end{equation}
\end{subequations}
Define the degree $-1$ linear map
\begin{subequations} \label{h def 0}
\begin{equation}
h_D : \Omega^{0,\bullet}\big( C, L_D \big) \longrightarrow \Omega^{0,\bullet-1}\big( C, L_D \big),
\end{equation}
whose only non-trivial component sends $\zeta = (\zeta_0, \zeta_\infty) 
\in \Omega^{0,1}(C, L_D)$ to $h_D(\zeta) = \big( h_D(\zeta)_0, h_D(\zeta)_\infty \big) \in \Omega^{0,0}(C, L_D)$, where
\begin{align}
\label{h def 0 b} h_D(\zeta)_0(z) &\coloneqq - \frac{1}{2 \pi \ii} \int_\mathbb{C} \frac{(1+ z \bar z')^{1+\deg(D)}}{(1 + |z'|^2)^{1+\deg(D)}} \frac{\d z' \wedge \zeta_0(z')}{z' - z}, \\
\label{h def 0 c} h_D(\zeta)_\infty(w) &\coloneqq - \frac{1}{2 \pi \ii} \int_\mathbb{C} \frac{(1+ w \bar w')^{1+\deg(D)}}{(1 + |w'|^2)^{1+\deg(D)}}  \frac{\d w' \wedge \zeta_\infty(w')}{w' - w}
\end{align}
\end{subequations}
are defined with respect to the chosen local charts.
\begin{theo} \label{thm: explicit SDRs 0}
For a divisor $D : C \to \bbZ$ with $\deg(D) \geq 0$, we have 
a continuous strong deformation retract
\begin{equation}
\begin{tikzcd}
\bbC^{1+\deg(D)} \arrow[r, "i_D", shift left=1mm] & \Omega^{0, \bullet}(C, L_D) \arrow[l, "p_D", shift left=1mm] \arrow["h_D"', out=-15,in=15,distance=8mm]
\end{tikzcd}
\end{equation}
with $i_D$, $p_D$ and $h_D$ defined in \eqref{i def 0}, \eqref{p def 0} and \eqref{h def 0}, respectively.
\end{theo}

It will be convenient to have an equivalent description of the strong deformation 
retract in Theorem \ref{thm: explicit SDRs 0} using a different model for the 
cohomology of $\Omega^{0, \bullet}(C, L_D)$. Explicitly, for a divisor 
$D : C \to \bbZ$ with $\deg(D) \geq 0$, consider the space of $D$-conditioned meromorphic functions
\begin{equation}
\mathcal O_D(C) \coloneqq \big\{ f \in \mathcal M(C) \,|\, (f) \geq - D \big\}.
\end{equation}
We then have an isomorphism
\begin{subequations} \label{ED def}
\begin{equation}
E_D \,:\, \bbC^{1+\deg(D)} \overset{\cong}\longrightarrow \mathcal O_D(C)
\end{equation}
sending $\alpha \in \bbC^{1+\deg(D)}$ to the $D$-conditioned meromorphic function $E_D(\alpha) \in \mathcal O_D(C)$ given by
\begin{equation} \label{ED def U0}
E_D(\alpha)(p) \coloneqq \frac{i_D(\alpha)_0(p)}{\psi^D_0(p)},
\end{equation}
for $p\in U_0$, and by
\begin{equation} \label{ED def Uinf}
E_D(\alpha)(q) \coloneqq \frac{i_D(\alpha)_\infty(q)}{\psi^D_\infty(q)},
\end{equation}
\end{subequations}
for $q\in U_\infty$. The two expressions in \eqref{ED def U0} and \eqref{ED def Uinf} 
agree on the overlap $U_{0\infty} = U_0 \cap U_\infty$ since 
$i_D(\alpha) = \big( i_D(\alpha)_0, i_D(\alpha)_\infty \big) \in \Omega^{0,\bullet}(C, L_D)$ 
and using the second relation in \eqref{transition function}.
\sk

Define the cochain map
\begin{equation} \label{i def 0 bis}
i'_D : \mathcal O_D(C) \longrightarrow \Omega^{0, \bullet}(C, L_D)
\end{equation}
sending $f \in \mathcal O_D(C)$ to $\big( \psi^D_0 f, \psi^D_\infty f \big) \in \Omega^{0, 0}(C, L_D)$. 
Equivalently, in view of the isomorphism \eqref{Gamma O iso}, we can equally describe this 
map as the canonical inclusion $i'_D : \mathcal O_D(C) \hookrightarrow \Gamma_D^{0,\bullet}(C)$. 
Alternatively, using the isomorphism \eqref{ED def}, we have $i'_D = i_D \circ E_D^{-1}$.
Define also the cochain map
\begin{equation} \label{p def 0 bis}
p'_D : \Omega^{0, \bullet}(C, L_D) \longrightarrow \mathcal O_D(C),
\end{equation}
given as $p'_D \coloneqq E_D \circ p_D$ or explicitly using \eqref{p def 0 b} as
\begin{equation} \label{p def 0 b bis}
p'_D(s)(z) \coloneqq - \frac{1 + \deg(D)}{2 \pi \ii \, \psi^D_0(z)} \int_\bbC \frac{s_0(z') (1+ z \bar z')^{\deg(D)}}{(1+|z'|^2)^{2 + \deg(D)}} \, \d z' \wedge \d \bar z'.
\end{equation}
The following is then an immediate consequence of Theorem \ref{thm: explicit SDRs 0}.
\begin{theobis}{thm: explicit SDRs 0} \label{thm: explicit SDRs 0 bis}
For a divisor $D : C \to \bbZ$ with $\deg(D) \geq 0$, we have a continuous strong deformation retract
\begin{equation}
\begin{tikzcd}
\mathcal O_D(C) \arrow[r, "i'_D", shift left=1mm] & \Omega^{0, \bullet}(C, L_D) \arrow[l, "p'_D", shift left=1mm] \arrow["h_D"', out=-15,in=15,distance=8mm]
\end{tikzcd}
\end{equation}
with $i'_D$, $p'_D$ and $h_D$ defined in \eqref{i def 0 bis}, \eqref{p def 0 bis} and \eqref{h def 0}, respectively.
\end{theobis}

\subsection{Case \texorpdfstring{$\deg(D) = -1$}{deg D = -1}} \label{sec: maps i p h deg -1}
The cohomology of $\Omega^{0, \bullet}(C, L_D)$ is trivial so that both of the 
maps $i_D : 0 \to \Omega^{0, \bullet}(C, L_D)$ and $p_D : \Omega^{0, \bullet}(C, L_D) \to 0$ are trivial.
Define the degree $-1$ linear map
\begin{subequations} \label{h def -1}
\begin{equation} \label{h def -1 a}
h_D : \Omega^{0,\bullet}\big( C, L_D \big) \longrightarrow \Omega^{0,\bullet-1}\big( C, L_D \big)
\end{equation}
whose only non-trivial component sends $\zeta = (\zeta_0, \zeta_\infty)
\in \Omega^{0,1}(C, L_D)$ to $h_D(\zeta) = \big( h_D(\zeta)_0, h_D(\zeta)_\infty \big)  \in \Omega^{0,0}(C, L_D)$, where
\begin{align}
\label{h def -1 b} h_D(\zeta)_0(z) &\coloneqq - \frac{1}{2 \pi \ii} \int_\mathbb{C} \frac{\d z' \wedge \zeta_0(z')}{z' - z}, \\
\label{h def -1 c} h_D(\zeta)_\infty(w) &\coloneqq - \frac{1}{2 \pi \ii} \int_\mathbb{C} \frac{\d w' \wedge \zeta_\infty(w')}{w' - w}
\end{align}
\end{subequations}
are defined with respect to the chosen local charts.
\begin{theo} \label{thm: explicit SDRs -1}
For a divisor $D : C \to \bbZ$ with $\deg(D) = -1$, we have a continuous strong deformation retract
\begin{equation}
\begin{tikzcd}
0 \arrow[r, "i_D", shift left=1mm] & \Omega^{0, \bullet}(C, L_D) \arrow[l, "p_D", shift left=1mm] \arrow["h_D"', out=-15,in=15,distance=8mm]
\end{tikzcd}
\end{equation}
with $h_D$ defined in \eqref{h def -1}.
\end{theo}

\subsection{Case \texorpdfstring{$\deg(D) \leq -2$}{deg D leq -2}} \label{sec: maps i p h deg neg}
The cohomology of $\Omega^{0, \bullet}(C, L_D)$ for $\deg(D) \leq -2$ is 
concentrated in degree $1$ and is isomorphic to $\bbC^{-1-\deg(D)}$.
Define the cochain map
\begin{subequations} \label{i def -2}
\begin{equation}
i_D : \bbC^{-1-\deg(D)}[-1] \longrightarrow \Omega^{0, \bullet}(C, L_D)
\end{equation}
sending $\alpha = ( \alpha_j )_{j=0}^{-2-\deg(D)} \in \bbC^{-1-\deg(D)}$ 
to $i_D(\alpha) = (i_D(\alpha)_0, i_D(\alpha)_\infty) \in \Omega^{0,1}(C, L_D)$, where
\begin{align}
i_D(\alpha)_0(z) &:= (1 + |z|^2)^{\deg(D)} \sum_{j=0}^{-2-\deg(D)} \alpha_j \bar z^j \d \bar z, \\
i_D(\alpha)_\infty(w) &:= - (1 + |w|^2)^{\deg(D)} \sum_{j=0}^{-2-\deg(D)} \alpha_j \bar w^{-2-\deg(D)-j} \d \bar w
\end{align}
\end{subequations}
are defined with respect to the chosen local charts.
Define the cochain map
\begin{subequations} \label{p def -2}
\begin{equation}
p_D : \Omega^{0, \bullet}(C, L_D) \longrightarrow \bbC^{-1-\deg(D)}[-1]
\end{equation}
to act trivially on $\Omega^{0, 0}(C, L_D)$ and by sending 
$\zeta = (\zeta_0, \zeta_\infty) \in \Omega^{0, 1}(C, L_D)$ to the coefficients 
$\big( p_D(\zeta)_j \big)_{j=0}^{-2-\deg(D)}$ of the polynomial
\begin{equation} \label{p def -2 b}
\sum_{j=0}^{-2-\deg(D)} p_D(\zeta)_j \bar z^j \coloneqq \frac{1+\deg(D)}{2 \pi \ii} \int_\bbC (1 + \bar z z')^{-2-\deg(D)} \d z' \wedge \zeta_0(z').
\end{equation}
\end{subequations}
Define the degree $-1$ linear map
\begin{subequations} \label{h def -2}
\begin{equation}
h_D : \Omega^{0,\bullet}\big( C, L_D \big) \longrightarrow \Omega^{0,\bullet-1}\big( C, L_D \big),
\end{equation}
whose only non-trivial component sends $\zeta = (\zeta_0, \zeta_\infty) 
\in \Omega^{0,1}(C, L_D)$ to $h_D(\zeta) = \big( h_D(\zeta)_0, h_D(\zeta)_\infty \big)  \in \Omega^{0,0}(C, L_D)$, where
\begin{align}
\label{h def -2 b} h_D(\zeta)_0(z) &\coloneqq - \frac{1}{2 \pi \ii} \int_\mathbb{C} \frac{(1+ \bar z z')^{-1-\deg(D)}}{(1+|z|^2)^{-1-\deg(D)}} \frac{\d z' \wedge \zeta_0(z')}{z' - z}, \\
\label{h def -2 c} h_D(\zeta)_\infty(w) &\coloneqq - \frac{1}{2 \pi \ii} \int_\mathbb{C} \frac{(1+ \bar w w')^{-1-\deg(D)}}{(1+|w|^2)^{-1-\deg(D)}} \frac{\d w' \wedge \zeta_\infty(w')}{w' - w}
\end{align}
\end{subequations}
are defined with respect to the chosen local charts.
\begin{theo} \label{thm: explicit SDRs -2}
For a divisor $D : C \to \bbZ$ with $\deg(D) \leq -2$ we have a 
continuous strong deformation retract
\begin{equation}
\begin{tikzcd}
\bbC^{-1-\deg(D)}[-1] \arrow[r, "i_D", shift left=1mm] & \Omega^{0, \bullet}(C, L_D) \arrow[l, "p_D", shift left=1mm] \arrow["h_D"', out=-15,in=15,distance=8mm]
\end{tikzcd}
\end{equation}
with $i_D$, $p_D$ and $h_D$ defined in \eqref{i def -2}, \eqref{p def -2} and \eqref{h def -2}, respectively.
\end{theo}

\subsection{Compatible pairs}
Let $\omega \in \mathcal M^{(1)}(C) \setminus \{0\}$ be a meromorphic $1$-form so that $\deg (\omega) = -2$.
\sk

Let $D : C \to \bbZ$ be a divisor and define the complementary divisor 
$D' \coloneqq (\omega) - D$. We have the non-degenerate degree $-1$ pairing
\begin{equation} \label{deg -1 pairing}
\langle\!\langle \cdot, \cdot \rangle\!\rangle_\omega : \Omega^{0,\bullet}(C, L_{D'}) \otimes \Omega^{0,\bullet}(C, L_D) \longrightarrow \bbC, \qquad
\zeta \otimes \alpha \longmapsto \frac{\ii}{2 \pi}\int_C \omega \wedge \zeta \wedge \alpha.
\end{equation}
\begin{lem} \label{lem: compatible pairs}
The strong deformation retracts constructed in Theorems \ref{thm: explicit SDRs 0}, 
\ref{thm: explicit SDRs -1} and \ref{thm: explicit SDRs -2} are compatible with the 
pairing \eqref{deg -1 pairing} in the sense that:
\begin{itemize}
  \item[(1)] $\langle\!\langle \zeta, \alpha \rangle\!\rangle_\omega = 0$ for all 
  $\zeta \in \mathrm{im}(i_{D'})$ and $\alpha \in \mathrm{ker}(p_D)$,
  \item[(2)] $\langle\!\langle \zeta, \alpha \rangle\!\rangle_\omega = 0$ for all 
  $\zeta \in \mathrm{ker}(p_{D'})$ and $\alpha \in \mathrm{im}(i_D)$, and
  \item[(3)] $\langle\!\langle h_{D'}(\zeta), \alpha \rangle\!\rangle_\omega = 
  (-1)^{|\zeta|} \langle\!\langle \zeta, h_D(\alpha) \rangle\!\rangle_\omega$ for 
  all $\zeta \in \Omega^{0,\bullet}(C, L_{D'})$ and $\alpha \in \Omega^{0,\bullet}(C, L_D)$.
\end{itemize}
\end{lem}


\section{\label{sec:application}Application to \texorpdfstring{$4$}{4}-dimensional Chern-Simons theory}
In this section we will use the strong deformation retracts from Section
\ref{sec:explicitSDR} to explicitly compute the Maurer-Cartan action and Lax connection 
for the $2$-dimensional integrable field theory associated with the cyclic $L_\infty$-algebra
$(\F(\Sigma),\ell')$, obtained by homotopy transfer from
the cyclic $L_\infty$-algebra $(\F(X), \ell)$ of $4$-dimensional Chern-Simons
theory with singularities and boundary conditions, as recalled in Section
\ref{sec:prelim} from \cite{BSV2}. For illustrative purposes, we will make a concrete choice of
meromorphic $1$-form $\omega$, singularities and boundary 
conditions for the purpose of recovering the principal chiral model with a 
WZ-term from the transferred $L_\infty$-algebra $(\F(\Sigma),\ell')$.

\subsection{Setup}
For concreteness, we will focus on the case where the meromorphic $1$-form 
$\omega \in \mathcal M^{(1)}(C) \setminus \{ 0 \}$ of $4$-dimensional Chern-Simons theory is given by
\begin{equation} \label{omega PCM+WZ def}
\omega = \frac{(\alpha-\widetilde{\alpha})^2(z-1)(z+1)}{(z-\alpha)^2 (z-\widetilde{\alpha})^2} \d z,
\end{equation}
where $\alpha, \widetilde{\alpha} \in \bbC \setminus \{0, 1, -1\}$ and $\alpha \neq \widetilde{\alpha}$.
On the $L_\infty$-algebra $(\mathcal E(X),\ell)$ of $4$-dimensional Chern-Simons theory we impose 
singularities by splitting the non-negative divisor of zeros as 
\begin{equation}
(\omega)_0 = (\omega)_0^+ + (\omega)_0^-, \qquad (\omega)^\pm_0 = \delta_{\pm 1},
\end{equation}
and local boundary conditions given by the choice of non-positive divisors
\begin{equation}
(\omega)_\infty^0 = (\omega)_\infty^\pm = (\omega)_\infty^2 = - \delta_\alpha - \delta_{\widetilde\alpha}.
\end{equation}
It will be convenient to introduce the combination
\begin{equation} \label{k def}
\kay \coloneqq \frac{\alpha \widetilde\alpha - 1}{\alpha - \widetilde\alpha}\,\in\,\bbC.
\end{equation}

\subsubsection{Strong deformation retract for $\mathcal L(X)$} \label{SDR for LX}
We use the results of Section \ref{sec:explicitSDR} to set up an explicit strong 
deformation retract from the cochain complex $\mathcal L(X)$ in \eqref{L X def} to its partial cohomology along $C$.
\sk

We can apply Theorem \ref{thm: explicit SDRs 0 bis} to the divisors $0$, $(\omega)_0^\pm$ 
and $(\omega)_0$, all of which are of non-negative degree. Taking the completed tensor
product of the resulting continuous strong deformation retracts with
$\Omega^0(\Sigma, \g)$, $\Omega^\pm(\Sigma, \g)$ and $\Omega^2(\Sigma, \g)$,
respectively, leads to the horizontal strong deformation retracts displayed below:
\begin{equation} \label{all SDRs for 4dCS Lax}
\begin{tikzcd}
\Omega^0(\Sigma, \g) \arrow[r, "i^0", shift left=1mm] \arrow[d, "\d^\pm_\Sigma"] & \g \otimes \Omega^{0, (0, \bullet)}(X) \arrow[l, "p^0", shift left=1mm] \arrow["h^0"', out=-10,in=10,distance=8mm] \arrow[d, "\d^\pm_\Sigma"]\\
\Omega^\pm (\Sigma, \g) \otimes \mathcal O_{(\omega)^\pm_0}(C) \arrow[r, "i^\pm", shift left=1mm] \arrow[d, "\d^\pm_\Sigma"] & \g \otimes \Omega^{\pm, (0, \bullet)}\big(X, L_{(\omega)_0^\pm} \big) \arrow[l, "p^\pm", shift left=1mm] \arrow["h^\pm"', out=-11.75,in=10,distance=8mm] \arrow[d, "\d_\Sigma^\mp"]\\
\Omega^2 (\Sigma, \g) \otimes \mathcal O_{(\omega)_0}(C) \arrow[r, "i^2", shift left=1mm] & \g \otimes \Omega^{2, (0, \bullet)}\big(X, L_{(\omega)_0} \big) \arrow[l, "p^2", shift left=1mm] \arrow["h^2"', out=-10,in=10,distance=8mm]
\end{tikzcd}
\end{equation}
The vertical maps on the left hand side are obtained by tensoring the 
light-cone components of the de Rham differentials along $\Sigma$ with the 
inclusions of $D$-conditioned meromorphic functions induced by the divisor 
inequalities $0\leq (\omega)^\pm_0\leq (\omega)_0$. Similarly, the vertical 
maps on the right hand side are obtained by tensoring with the corresponding 
holomorphic line bundle morphisms. Explicitly, we have
\begin{subequations}\label{d pm Lax}
\begin{align} \label{d pm Sigma Lax}
\d_\Sigma^\pm \,:\, \Omega^{0, (0, \bullet)}(X) &\longrightarrow \Omega^{\pm, (0, \bullet)}\big(X, L_{(\omega)_0^\pm} \big), \notag\\
(s_0, s_\infty) &\longmapsto \big( (z \mp 1) \, \d_\Sigma^\pm s_0, (1 \mp w) \, \d_\Sigma^\pm s_\infty \big), \\
\label{d pm Sigma Lax 2}
\d_\Sigma^\mp \,:\, \Omega^{\pm, (0, \bullet)}\big(X, L_{(\omega)_0^\pm }\big) &\longrightarrow \Omega^{2, (0, \bullet)}\big(X, L_{(\omega)_0}\big), \notag\\
(s_0, s_\infty) &\longmapsto \big( (z \pm 1) \, \d_\Sigma^\mp s_0, (1 \pm w) \, \d_\Sigma^\mp s_\infty \big).
\end{align}
\end{subequations}
We now apply the totalization construction of \cite[Appendix A.1]{BSV2} to the 
horizontal strong deformation retracts from \eqref{all SDRs for 4dCS Lax}, treating 
the vertical differential on the right hand side as a small perturbation in the homological 
perturbation lemma. Noting that the simplification from \cite[Remark A.1]{BSV2} applies in 
the present case since the horizontal $i$ maps in the diagram \eqref{all SDRs for 4dCS Lax} 
commute with the vertical differentials $\d^\pm_\Sigma$, we obtain the perturbed strong deformation retract
\begin{subequations} \label{SDR for Lax complex}
\begin{equation}
\begin{tikzcd}
\big( \L(\Sigma), \ell'_1 \big) \arrow[r, "\widetilde{i}", shift left=1mm] & \big( \L(X), \d_\Sigma + \bar\partial \big) \arrow[l, "\widetilde{p}", shift left=1mm] \arrow["\widetilde{h}"', out=-15,in=15,distance=8mm]
\end{tikzcd}
\end{equation}
with perturbed differential $\ell'_1 \coloneqq \d_\Sigma$ and strong deformation 
retract data $\widetilde{i} \coloneqq i$, $\widetilde{p} \coloneqq p + p\, \d_\Sigma\, h$ 
and $\widetilde{h} \coloneqq h$. Note that the contribution $p\, \d_\Sigma\, h$ to the 
component $\widetilde{p}^0$ of $\widetilde{p}$ vanishes on degree grounds.
Therefore, the non-vanishing components of $\ell'_1, \widetilde{i}, \widetilde{p}, \widetilde{h}$ are
\begin{align}
\label{map tilde d Lax} \ell'_1 &= \d_\Sigma,\\
\label{map tilde i Lax} \widetilde{i}^0 &= i^0, & \widetilde{i}^1 &= i^+ + i^-, & \widetilde{i}^2 &= i^2\\
\label{map tilde p Lax} \widetilde{p}^0 &= p^0, & \widetilde{p}^1 &= p^+ + p^- + (p^+ \d^+_\Sigma + p^- \d^-_\Sigma) h^0, & \widetilde{p}^2 &= p^2 + p^2 (\d^-_\Sigma h^+ + \d^+_\Sigma h^-)\\
\label{map tilde h Lax} \widetilde{h}^1 &= h^0, & \widetilde{h}^2 &= h^+ + h^-, & \widetilde{h}^3 &= h^2.
\end{align}
\end{subequations}
In order to help parse the information contained in the strong deformation retract
\eqref{SDR for Lax complex}, it is convenient to display it in full (with cohomological 
degrees extending vertically from $0$ to $3$):
\begin{equation}
\begin{tikzcd}
\Omega^0(\Sigma, \g) \arrow[d, "\d_\Sigma"] \arrow[r, "i^0", shift left=1mm] & \g \otimes \Omega^{0, (0, 0)}(X) \arrow[l, "p^0", shift left=1mm] \arrow[d, "\d_\Sigma + \bar\partial", shift left=1mm]\\
{\displaystyle \bigoplus_{i=\pm}}\; \Omega^i (\Sigma, \g) \otimes \mathcal O_{(\omega)^i_0}(C) \arrow[d, "\d_\Sigma"] \arrow[r, "i^+ + i^-", shift left=1mm] & {\displaystyle \bigoplus_{i = \pm}} \; \g \otimes \Omega^{i, (0, 0)}\big(X, L_{(\omega)_0^i}\big) \oplus \Omega^{0, (0, 1)}(X) \arrow[l, "\widetilde{p}^1", shift left=1mm] \arrow[d, "\d_\Sigma + \bar\partial", shift left=1mm] \arrow[u, "h^0", shift left=1mm]\\
\Omega^2(\Sigma, \g) \otimes \mathcal O_{(\omega)_0}(C) \arrow[r, "i^2", shift left=1mm] \arrow[d, "0"] & {\displaystyle \bigoplus_{i = \pm}} \; \g \otimes \Omega^{i, (0, 1)}\big(X, L_{(\omega)_0^i}\big) \oplus \g \otimes \Omega^{2, (0, 0)}\big(X, L_{(\omega)_0}\big) \arrow[l, "\widetilde{p}^2", shift left=1mm] \arrow[d, "\d_\Sigma + \bar\partial", shift left=1mm] \arrow[u, "h^+ + h^-", shift left=1mm]\\
0 \arrow[r, "0", shift left=1mm] & \g \otimes \Omega^{2, (0, 1)}\big(X, L_{(\omega)_0}\big) \arrow[l, "0", shift left=1mm] \arrow[u, "h^2", shift left=1mm]
\end{tikzcd}
\end{equation}

\subsubsection{Strong deformation retract for $\mathcal F(X)$} \label{SDR for FX}
Next, we also use the results of Section \ref{sec:explicitSDR} to set up an explicit 
strong deformation retract from the cochain complex $\mathcal F(X)$ in \eqref{F X def} to its partial cohomology along $C$.
\sk

We can apply Theorems \ref{thm: explicit SDRs -2}, \ref{thm: explicit SDRs -1}
and \ref{thm: explicit SDRs 0} to the divisors $(\omega)^0_\infty$,
$(\omega)_0^\pm + (\omega)^\pm_\infty$ and $(\omega)_0 + (\omega)^2_\infty$
of degrees $-2$, $-1$ and $0$, respectively. Then taking the completed tensor
product of the resulting continuous strong deformation retracts with
$\Omega^0(\Sigma, \g)$, $\Omega^\pm(\Sigma, \g)$ and $\Omega^2(\Sigma, \g)$,
respectively, leads to the horizontal strong deformation retracts displayed below:
\begin{equation} \label{all SDRs for 4dCS}
\begin{tikzcd}
\Omega^0(\Sigma, \g)[-1] \arrow[r, "i^0", shift left=1mm] & \g \otimes \Omega^{0, (0, \bullet)}\big(X, L_{(\omega)^0_\infty}\big) \arrow[l, "p^0", shift left=1mm] \arrow["h^0"', out=-10,in=10,distance=8mm] \arrow[d, "\d^\pm_\Sigma"]\\
\qquad\qquad 0 \quad \arrow[r, "i^\pm", shift left=1mm] & \g \otimes \Omega^{\pm, (0, \bullet)}\big(X, L_{(\omega)_0^\pm + (\omega)^\pm_\infty}\big) \arrow[l, "p^\pm", shift left=1mm] \arrow["h^\pm"', out=-11.75,in=10,distance=8mm] \arrow[d, "\d_\Sigma^\mp"]\\
\Omega^2(\Sigma, \g) \arrow[r, "i^2", shift left=1mm] & \g \otimes \Omega^{2, (0, \bullet)}\big(X, L_{(\omega)_0 + (\omega)^2_\infty}\big) \arrow[l, "p^2", shift left=1mm] \arrow["h^2"', out=-10,in=10,distance=8mm]
\end{tikzcd}
\end{equation}
The top strong deformation retract is defined as in Section \ref{sec: maps i p h deg neg}
with degree $-2$ divisor given by $(\omega)^0_\infty$, the middle strong
deformation retract is defined as in Section \ref{sec: maps i p h deg -1}
with degree $-1$ divisor given by $(\omega)_0^\pm + (\omega)^\pm_\infty$
and the bottom strong deformation retract is defined as in Section 
\ref{sec: maps i p h deg pos} with degree $0$ divisor given by $(\omega)_0 + (\omega)^2_\infty$.
The vertical maps on the right hand side of \eqref{all SDRs for 4dCS} are given
by the tensor product of the light-cone components of the de Rham differentials
with corresponding holomorphic line bundle morphisms:
\begin{subequations}\label{d pm}
\begin{align} \label{d pm Sigma}
\d_\Sigma^\pm \,:\, \g \otimes \Omega^{0, (0, \bullet)}\big(X, L_{(\omega)^0_\infty}\big) &\longrightarrow \g \otimes \Omega^{\pm, (0, \bullet)}\big(X, L_{(\omega)_0^\pm + (\omega)^\pm_\infty}\big), \notag\\
(s_0, s_\infty) &\longmapsto \big( (z \mp 1) \, \d_\Sigma^\pm s_0, (1 \mp w) \, \d_\Sigma^\pm s_\infty \big), \\
\label{d pm Sigma 2}
\d_\Sigma^\mp \,:\, \g \otimes \Omega^{\pm, (0, \bullet)}\big(X, L_{(\omega)_0^\pm + (\omega)^\pm_\infty}\big) &\longrightarrow \g \otimes \Omega^{2, (0, \bullet)}\big(X, L_{(\omega)_0 + (\omega)^2_\infty}\big), \notag\\
(s_0, s_\infty) &\longmapsto \big( (z \pm 1) \, \d_\Sigma^\mp s_0, (1 \pm w) \, \d_\Sigma^\mp s_\infty \big).
\end{align}
\end{subequations}
Applying the totalization construction of \cite[Appendix A.1]{BSV2} to the
horizontal strong deformation retracts from \eqref{all SDRs for 4dCS}
and treating the vertical differential on the right hand side as a small
perturbation in the homological perturbation lemma, we obtain the perturbed strong deformation retract
\begin{subequations} \label{SDR for field complex}
\begin{equation}
\begin{tikzcd}
\big( \F(\Sigma), \ell'_1 \big) \arrow[r, "\widetilde{i}", shift left=1mm] & \big( \F(X), \d_\Sigma + \bar\partial \big) \arrow[l, "\widetilde{p}", shift left=1mm] \arrow["\widetilde{h}"', out=-15,in=15,distance=8mm]
\end{tikzcd}
\end{equation}
with perturbed differential $\ell'_1 \coloneqq p \, \d_\Sigma\, i + p \, \d_\Sigma \, h \, \d_\Sigma\, i$
and strong deformation retract data $\widetilde{i} \coloneqq i + h\, \d_\Sigma\, i$,
$\widetilde{p} \coloneqq p + p\, \d_\Sigma\, h$ and $\widetilde{h} \coloneqq h$.
Note that here only contributions up to first order in $h$ enter, 
because the Dolbeault complexes over $\CP$ are concentrated only in degrees $0$ and $1$ 
and the perturbation $\d_\Sigma$ preserves the Dolbeault degree. 
Furthermore, it follows from \eqref{all SDRs for 4dCS} that $\ell'_1$
has vanishing zeroth order contributions $p^\pm \, \d^\pm_\Sigma \, i^0 = 0$ and
$p^2 \, \d^\mp_\Sigma \, i^\pm = 0$ because $p^\pm = 0$ and $i^\pm = 0$.
Similarly, $\widetilde{i}^2$ and $\widetilde{p}^1$ have vanishing first order
contributions because $\d_\Sigma\, i^2 = 0$ and $p^\pm = 0$. Summing up, 
the non-vanishing components of $\ell'_1, \widetilde{i}, \widetilde{p}, \widetilde{h}$ are
\begin{align}
\label{ell'1 def} \ell'_1 &= p^2 \d_\Sigma (h^+ \d^+_\Sigma + h^- \d^-_\Sigma) i^0,\\
\label{tilde i def} \widetilde{i}^1 &= i^0 + (h^+ \d^+_\Sigma + h^- \d^-_\Sigma) i^0, & \widetilde{i}^2 &= i^2,\\
\label{tilde p def} \widetilde{p}^1 &= p^0, & \widetilde{p}^2 &= p^2 + p^2 (\d^-_\Sigma h^+ + \d^+_\Sigma h^-),\\
\widetilde{h}^1 &= h^0, & \widetilde{h}^2 &= h^+ + h^-, & \widetilde{h}^3 &= h^2.
\end{align}
\end{subequations}
Note that, by \eqref{all SDRs for 4dCS}
the zeroth order contribution $p\, \d_\Sigma\, i = 0$ to $\ell'_1$ and the
first order contributions $h\, \d_\Sigma\, i = 0$ to $\widetilde{i}^2$ and
$p\, \d_\Sigma\, h = 0$ to $\widetilde{p}^1$ all vanish.
As in Section \ref{SDR for LX}, to help parse the information contained in the strong deformation retract
\eqref{SDR for field complex} we display it in full below (with cohomological 
degrees extending vertically from $0$ to $3$):
\begin{equation} \label{SDR written in full}
\begin{tikzcd}
0 \arrow[d, "0"] \arrow[r, "0", shift left=1mm] & \g \otimes \Omega^{0, (0, 0)}\big(X, L_{(\omega)^0_\infty}\big) \arrow[l, "0", shift left=1mm] \arrow[d, "\d_\Sigma + \bar\partial", shift left=1mm]\\
\Omega^0(\Sigma, \g) \arrow[d, "\ell'_1"] \arrow[r, "\widetilde{i}^1", shift left=1mm] & {\displaystyle \bigoplus_{i = \pm}} \; \g \otimes \Omega^{i, (0, 0)}\big(X, L_{(\omega)_0^i + (\omega)^i_\infty}\big) \oplus \g \otimes \Omega^{0, (0, 1)}\big(X, L_{(\omega)^0_\infty}\big) \arrow[l, "p^0", shift left=1mm] \arrow[d, "\d_\Sigma + \bar\partial", shift left=1mm] \arrow[u, "h^0", shift left=1mm]\\
\Omega^2(\Sigma, \g) \arrow[d, "0"] \arrow[r, "i^2", shift left=1mm]
& \g \otimes \Omega^{2, (0, 0)}\big(X, L_{(\omega)_0 + (\omega)^2_\infty}\big) \oplus {\displaystyle \bigoplus_{i = \pm}} \; \g \otimes \Omega^{i, (0, 1)}\big(X, L_{(\omega)_0^i + (\omega)^i_\infty}\big) \arrow[l, "\widetilde{p}^2", shift left=1mm] \arrow[d, "\d_\Sigma + \bar\partial", shift left=1mm] \arrow[u, "h^+ + h^-", shift left=1mm]\\
0 \arrow[r, "0", shift left=1mm] & \g \otimes \Omega^{2, (0, 1)}\big(X, L_{(\omega)_0 + (\omega)^2_\infty}\big) \arrow[l, "0", shift left=1mm] \arrow[u, "h^2", shift left=1mm]
\end{tikzcd}
\end{equation}

\subsection{Field complex} \label{sec: ell' 1 differential}
In this section we explicitly compute the differential \eqref{ell'1 def}
of the field complex $\F(\Sigma)$.
\sk

Let $\phi \in \F(\Sigma)^1 = \Omega^0(\Sigma, \g)$. Applying
$i^0$ from \eqref{all SDRs for 4dCS} (see also \eqref{i def -2} for the
explicit formula) we obtain
\begin{equation} \label{tau in application}
i^0(\phi) = \bigg( \phi \, \frac{\d \bar z}{(1+|z|^2)^2},
- \phi \, \frac{\d \bar w}{(1+|w|^2)^2} \bigg).
\end{equation}
Applying further the perturbation differential \eqref{d pm Sigma} gives
\begin{equation} \label{d i0 tau}
\d^\pm_\Sigma i^0 (\phi) = \bigg( \d^\pm_\Sigma \phi \, \frac{(z \mp 1) \, \d \bar z}{(1+|z|^2)^2},
- \d^\pm_\Sigma \phi \, \frac{(1 \mp w) \, \d \bar w}{(1+|w|^2)^2} \bigg).
\end{equation}
Next, we apply $h^\pm$ from \eqref{all SDRs for 4dCS} (see also 
\eqref{h def -1 b} and \eqref{h def -1 c} for the explicit formula) and find
\begin{align} \label{hdi phi}
h^\pm \d^\pm_\Sigma i^0 (\phi) &=
\bigg( \frac{1}{2 \pi \ii}\, \d^\pm_\Sigma \phi \int_\bbC \frac{(z' \mp 1) \, \d z' \wedge \d \bar z'}{(1 + |z'|^2)^2 (z' - z)},
- \frac{1}{2 \pi \ii}\, \d^\pm_\Sigma \phi \int_\bbC \frac{(1 \mp w') \, \d w' \wedge \d \bar w'}{(1 + |w'|^2)^2 (w'-w)} \bigg) \notag \\
&= \bigg( \d^\pm_\Sigma \phi \, \frac{\mp \bar z - 1}{1 + |z|^2},
\d^\pm_\Sigma \phi \, \frac{- \bar w \mp 1}{1 + |w|^2} \bigg),
\end{align}
where the first step takes into account the additional minus sign arising 
from the homotopy $h^\pm$ passing through the de Rham differentials
$\d^\pm_\Sigma \phi$ and the second step evaluates the integrals
using Corollary \ref{cor: del bar problem}.
Applying the perturbation differential \eqref{d pm Sigma 2} to \eqref{hdi phi}
and using the fact that $\d^-_\Sigma \d^+_\Sigma = - \d^+_\Sigma \d^-_\Sigma$ we obtain
\begin{equation} \label{d h d i on tau}
\d_\Sigma (h^+ \d^+_\Sigma + h^- \d^-_\Sigma) i^0 (\phi)
= \big( 2 \d^+_\Sigma \d^-_\Sigma \phi, 2 \d^+_\Sigma \d^-_\Sigma \phi \big).
\end{equation}
Finally, applying $p^2$ from \eqref{all SDRs for 4dCS} (see also \eqref{p def 0}
for the explicit formula) we find
\begin{equation} \label{eqn:ell'_1}
\ell'_1(\phi) = p^2 (\d_\Sigma^- h^+ \d^+_\Sigma + \d_\Sigma^+ h^- \d^-_\Sigma) i^0(\phi)
= 2 \d^+_\Sigma \d^-_\Sigma \phi \; \frac{\ii}{2 \pi} \int_\bbC \frac{\d z' \wedge \d \bar z'}{(1+|z'|^2)^2}
= 2 \d^+_\Sigma \d^-_\Sigma \phi,
\end{equation}
where the integral in the last step evaluates to $1$ using Lemma
\ref{lem: reproducing kernel} with $N=0$ and $P(z) = 1$. In summary, the field complex is therefore given explicitly by
\begin{equation} \label{field complex}
\begin{tikzcd}[column sep=15mm]
\big( \F(\Sigma), \ell'_1 \big) \coloneqq \bigg( \overset{(1)}{\Omega^0(\Sigma, \g)} \arrow[r, "2 \d^+_\Sigma \d^-_\Sigma"] & \overset{(2)}{\Omega^2(\Sigma, \g)} \bigg).
\end{tikzcd}
\end{equation}

\subsection{Maurer-Cartan action} \label{sec: deriving MC action}
In this section we will use the homotopy transfer theorem to compute the Maurer-Cartan action
\begin{equation}\label{eqn:MCaction}
S_{\rm MC}[\phi] = \pairsize{\phi}{\sum_{n \geq 1} \frac{1}{(n+1)!} 
\ell'_n(\phi^{\otimes n})}_{\omega}^\prime
\end{equation}
associated with the cyclic $L_\infty$-algebra structure
$(\ell^\prime_n)_{n \geq 2}, \pair{\cdot}{\cdot}_\omega^\prime$
on the field complex $(\F(\Sigma), \ell'_1)$ obtained via transfer 
from the simple cyclic $L_\infty$-algebra structure $\ell_2, \pair{\cdot}{\cdot}_\omega$
on Chern-Simons theory with singularities and boundary conditions
$(\F(X), \d_\Sigma + \bar\partial)$ along the
strong deformation retract \eqref{SDR for field complex}.
Here $\phi \in \Omega_\cc^0(\Sigma,\g) = \F_\cc(\Sigma)^1$ 
is a degree $1$ element of the compactly supported field complex.
\sk

In order to compute \eqref{eqn:MCaction}, it suffices to consider the action
\begin{subequations}\label{eqn:Lie}
\begin{equation}
\ell_2 \,:\, \F(X)^1 \otimes \F(X)^1 \,\longrightarrow\, \F(X)^2
\end{equation}
of the Lie bracket $\ell_2$ on $(\F(X), \d_\Sigma + \bar\partial)$
on degree $1$ elements $s = (s_0, s_\infty), t = (t_0, t_\infty) \in \F(X)^1$. The latter reads explicitly as
\begin{equation}
\ell_2(s,t) = \big( P^{\alpha,\widetilde\alpha}_0\, [s_0 \,\overset{\wedge},\, t_0] \, , \, P^{\alpha,\widetilde\alpha}_\infty \, [s_\infty \,\overset{\wedge},\, t_\infty] \big) \in \F(X)^2,
\end{equation}
where multiplication by the polynomials
\begin{equation} \label{P alpha poly def}
P^{\alpha,\widetilde\alpha}_0(z) \coloneqq \frac{(z-\alpha) (z - \widetilde\alpha)}{\alpha - \widetilde\alpha}, \qquad
P^{\alpha,\widetilde\alpha}_\infty(w) \coloneqq \frac{(1-\alpha w) (1 - \widetilde\alpha w)}{\alpha - \widetilde\alpha},
\end{equation}
\end{subequations}
which include the normalizing factors $(\alpha - \widetilde\alpha)^{-1}$ 
for later convenience, accounts at the level of local sections for the holomorphic line bundle morphisms
\begin{equation}
L_{(\omega)_0^\pm + (\omega)^\pm_\infty} \otimes L_{(\omega)^0_\infty} \longrightarrow L_{(\omega)_0^\pm + (\omega)^\pm_\infty} ,\qquad
L_{(\omega)_0^\pm + (\omega)^\pm_\infty} \otimes L_{(\omega)_0^\mp + (\omega)^\mp_\infty} \longrightarrow L_{(\omega)_0 + (\omega)^2_\infty}
\end{equation}
entering the action of the Lie bracket $\ell_2$ on degree $1$ elements of $\F(X)$,
see \eqref{SDR written in full}.
\sk

To determine \eqref{eqn:MCaction},
we are interested in evaluating, for all $n \geq 2$, the transferred $\ell'_n$-brackets
from \eqref{elln def} on $\phi^{\otimes n}$, where $\phi \in \Omega^0(\Sigma, \g) = \F(\Sigma)^1$.
In this case, both the Koszul sign \eqref{chi sign def} and the sign \eqref{zeta sign def}
are identically equal to $1$. Moreover, the terms in the sum over the shuffle
$\sigma \in \text{Sh}(m,n-m)$ become independent of $\sigma$ since we have
$a_{\sigma(i)} = \phi$ for all $i =1,\ldots, n$ and hence this sum simply
produces an overall factor of $|\text{Sh}(m,n-m)| = \binom{n}{m}$. Therefore
\eqref{in def} and \eqref{elln def} simplify, for all $n \geq 2$, to
\begin{align}
\label{in on phi} \widetilde{i}_n(\phi^{\otimes n}) &= \frac{1}{2} \sum_{m=1}^{n-1} \binom{n}{m} \, \widetilde{h} \, \ell_2\Big( \widetilde{i}_m(\phi^{\otimes m}), \widetilde{i}_{n-m}\big(\phi^{\otimes (n-m)} \big) \Big),\\
\label{elln on phi} \ell'_n(\phi^{\otimes n}) &= \frac{1}{2} \sum_{m=1}^{n-1} \binom{n}{m} \, \widetilde{p} \, \ell_2\Big( \widetilde{i}_m(\phi^{\otimes m}), \widetilde{i}_{n-m}\big( \phi^{\otimes (n-m)} \big) \Big),
\end{align}
where $\widetilde{i}_1 \coloneqq \widetilde{i}$.
In the case of \eqref{in on phi} there is yet a further simplification coming
from the fact that all maps $\widetilde{i}_n$, for $n \geq 2$, end with a homotopy
$\widetilde{h}$ and therefore they land in forms of Dolbeault degree $0$.
Since our homotopy $\widetilde{h}$ vanishes on forms of Dolbeault degree $0$
and $\ell_2$ preserves Dolbeault degree, it follows that in \eqref{in on phi}
the contributions to the sum vanish for $2 \leq m \leq n-2$. This leads to
$\widetilde{i}_2(\phi^{\otimes 2}) = \widetilde{h} \big( \ell_2\big( \widetilde{i}(\phi), \widetilde{i}(\phi) \big) \big)$
and, using also the graded antisymmetry of $\ell_2$,
$\widetilde{i}_n(\phi^{\otimes n}) = n \, \widetilde{h} \, 
\ell_2\big( \widetilde{i}(\phi), \widetilde{i}_{n-1}\big( \phi^{\otimes (n-1)} \big) \big)$,
for all $n \geq 3$. Therefore, for all $n \geq 2$, by recursion we conclude
\begin{equation}\label{in on phi simplified}
\widetilde{i}_n(\phi^{\otimes n}) = \frac{n!}2 \, \widetilde{h} \, \ell_2\Big( \widetilde{i}(\phi), \ldots \widetilde{h} \ell_2 \big( \widetilde{i}(\phi), \widetilde{i}(\phi) \big) \Big),
\end{equation}
where $\widetilde{h} \ell_2(\widetilde{i}(\phi),\cdot)$ is iterated $n-1$ times.
This expression can be represented graphically as a tree with a single long
branch extending to the right (or to the left by graded antisymmetry of $\ell_2$):
\begin{equation}
\widetilde{i}_n(\phi^{\otimes n}) \; = \; \frac{n!}2
\raisebox{-12.5mm}{
\begin{tikzpicture}[cir/.style={circle,draw=black,inner sep=0pt,minimum size=2mm},
        poin/.style={rectangle, inner sep=2pt,minimum size=0mm},scale=0.8, every node/.style={scale=0.8},
		every edge quotes/.style = {auto, font=\footnotesize}]

\node[poin] (i)   at (0,-4) {};
\node[poin] (v1)  at (0,-3) {$\ell_2$};
\node[poin] (v2)  at (-1,-2) {$\phi$};
\node[poin] (v3)  at (1,-2) {$\ell_2$};
\node[poin] (v4b) at (0,-1) {$\phi$};
\node[poin] (v34) at (1.5,-1.4) {$\iddots$};
\node[poin] (v4)  at (2,-1) {$\ell_2$};
\node[poin] (v5)  at (3,0) {$\phi$};
\node[poin] (v5b) at (1,0) {$\phi$};

\draw[thick] (i)  edge["$\widetilde{h}$"] (v1);
\draw[thick] (v1) edge["$\widetilde{i}$"] (v2);
\draw[thick] (v1) edge["$\widetilde{h}$"'] (v3);
\draw[thick] (v3) edge["$\widetilde{i}$"] (v4b);
\draw[thick] (v4) edge["$\widetilde{i}$"'] (v5);
\draw[thick] (v4) edge["$\widetilde{i}$"] (v5b);
\end{tikzpicture}
}
\; = \; \frac{n!}2
\raisebox{-12.5mm}{
\begin{tikzpicture}[cir/.style={circle,draw=black,inner sep=0pt,minimum size=2mm},
        poin/.style={rectangle, inner sep=2pt,minimum size=0mm},scale=0.8, every node/.style={scale=0.8},
		every edge quotes/.style = {auto, font=\footnotesize}]

\node[poin] (i)   at (0,-4) {};
\node[poin] (v1)  at (0,-3) {$\ell_2$};
\node[poin] (v2)  at (1,-2) {$\phi$};
\node[poin] (v3)  at (-1,-2) {$\ell_2$};
\node[poin] (v4b) at (0,-1) {$\phi$};
\node[poin] (v34) at (-1.5,-1.4) {$\ddots$};
\node[poin] (v4)  at (-2,-1) {$\ell_2$};
\node[poin] (v5)  at (-3,0) {$\phi$};
\node[poin] (v5b) at (-1,0) {$\phi$};

\draw[thick] (i)  edge["$\widetilde{h}$"] (v1);
\draw[thick] (v1) edge["$\widetilde{i}$"'] (v2);
\draw[thick] (v1) edge["$\widetilde{h}$"] (v3);
\draw[thick] (v3) edge["$\widetilde{i}$"'] (v4b);
\draw[thick] (v4) edge["$\widetilde{i}$"] (v5);
\draw[thick] (v4) edge["$\widetilde{i}$"'] (v5b);
\end{tikzpicture}
}
\end{equation}
It follows from the explicit form of the perturbed strong deformation retract \eqref{SDR for field complex} that
\begin{equation}
\widetilde{i}_n(\phi^{\otimes n}) \; = \; n!
\raisebox{-12.5mm}{
\begin{tikzpicture}[cir/.style={circle,draw=black,inner sep=0pt,minimum size=2mm},
        poin/.style={rectangle, inner sep=2pt,minimum size=0mm},scale=0.8, every node/.style={scale=0.8},
		every edge quotes/.style = {auto, font=\footnotesize}]

\node[poin] (i)   at (0,-4) {};
\node[poin] (v1)  at (0,-3) {$\ell_2$};
\node[poin] (v2)  at (-1,-2) {$\phi$};
\node[poin] (v3)  at (1,-2) {$\ell_2$};
\node[poin] (v4b) at (0,-1) {$\phi$};
\node[poin] (v34) at (1.5,-1.4) {$\iddots$};
\node[poin] (v4)  at (2,-1) {$\ell_2$};
\node[poin] (v5)  at (3,0) {$\phi$};
\node[poin] (v5b) at (1,0) {$\phi$};

\draw[thick] (i)  edge["$h^+$"] (v1);
\draw[thick] (v1) edge["$i^0$"] (v2);
\draw[thick] (v1) edge["$h^+$"'] (v3);
\draw[thick] (v3) edge["$i^0$"] (v4b);
\draw[thick] (v4) edge["$h^+ \d^+_\Sigma i^0$"'] (v5);
\draw[thick] (v4) edge["$i^0$"] (v5b);
\end{tikzpicture}
}
\;+\; n!
\raisebox{-12.5mm}{
\begin{tikzpicture}[cir/.style={circle,draw=black,inner sep=0pt,minimum size=2mm},
        poin/.style={rectangle, inner sep=2pt,minimum size=0mm},scale=0.8, every node/.style={scale=0.8},
		every edge quotes/.style = {auto, font=\footnotesize}]

\node[poin] (i)   at (0,-4) {};
\node[poin] (v1)  at (0,-3) {$\ell_2$};
\node[poin] (v2)  at (-1,-2) {$\phi$};
\node[poin] (v3)  at (1,-2) {$\ell_2$};
\node[poin] (v4b) at (0,-1) {$\phi$};
\node[poin] (v34) at (1.5,-1.4) {$\iddots$};
\node[poin] (v4)  at (2,-1) {$\ell_2$};
\node[poin] (v5)  at (3,0) {$\phi$};
\node[poin] (v5b) at (1,0) {$\phi$};

\draw[thick] (i)  edge["$h^-$"] (v1);
\draw[thick] (v1) edge["$i^0$"] (v2);
\draw[thick] (v1) edge["$h^-$"'] (v3);
\draw[thick] (v3) edge["$i^0$"] (v4b);
\draw[thick] (v4) edge["$h^- \d^-_\Sigma i^0$"'] (v5);
\draw[thick] (v4) edge["$i^0$"] (v5b);
\end{tikzpicture}
}
\end{equation}
For all $n \geq 2$, we can write out explicitly what the above diagrams represent as
\begin{equation} \label{in on phi explicit}
\widetilde{i}_n(\phi^{\otimes n}) = n! \, \big( (F^+_\phi)^{n-1} h^+ \d^+_\Sigma  
+ (F^-_\phi)^{n-1} h^- \d^-_\Sigma \big) \, i^0(\phi),
\end{equation}
where we introduced the linear maps
\begin{align} \label{linear map F def}
F^\pm_\phi \,:\, \g \otimes \Omega^{\pm, (0, 0)}\big(X, L_{(\omega)_0^\pm + (\omega)^\pm_\infty}\big) \,&\longrightarrow\, \g \otimes \Omega^{\pm, (0, 0)}\big(X, L_{(\omega)_0^\pm + (\omega)^\pm_\infty}\big), \notag\\
s \,& \longmapsto\, h^\pm \ell_2\big( i^0(\phi), s \big).
\end{align}
Furthermore, the evaluation $F^\pm_\phi(s^\pm) = (F^\pm_\phi(s^\pm)_0,F^\pm_\phi(s^\pm)_\infty)$
on a section $s^\pm = \big( s^\pm_0, s^\pm_\infty \big) \in \g \otimes
\Omega^{\pm, (0, 0)}\big(X, L_{(\omega)_0^\pm + (\omega)^\pm_\infty}\big)$
reads explicitly as follows:
\begin{subequations} \label{F beta explicit}
\begin{align} 
F^\pm_\phi(s^\pm)_0(z) &= \int_\mathbb{C} \omega_{\rm FS}(z')\, P^{\alpha, \widetilde\alpha}_0(z')\, \frac{[\phi, s^\pm_0(z')]}{z' - z}, \\
F^\pm_\phi(s^\pm)_\infty(w) &= - \int_\mathbb{C} \omega_{\rm FS}(w')\, P^{\alpha, \widetilde\alpha}_\infty
(w')\, \frac{[\phi, s^\pm_\infty(w')]}{w' - w}.
\end{align}
\end{subequations}
Here we used \eqref{tau in application} to evaluate $i_0$, 
\eqref{eqn:Lie} to evaluate $\ell_2$ and \eqref{h def -1} to evaluate $h^\pm$.
For notational convenience, we also introduced the Fubini-Study volume form
\begin{equation}\label{eqn:FS}
\omega_{\rm FS} = - \frac{1}{2 \pi \ii} \frac{\d z \wedge \d \bar z}{(1+|z|^2)^2}
\end{equation}
associated with the Fubini-Study metric on $\CP$.
\begin{lem} \label{lem: F identities}
For any $\phi \in \Omega^0(\Sigma, \g) = \F(\Sigma)^1$ and 
$s^\pm \in \g \otimes \Omega^{\pm, (0, 0)}\big(X, L_{(\omega)_0^\pm + (\omega)^\pm_\infty}\big) \subset \F(X)^1$ 
we have
\begin{align}
\label{F identity a} \big\langle \phi, p^2 \ell_2 \big( F^+_\phi(s^+), s^- \big) \big\rangle
&= \big\langle \phi, p^2 \ell_2 \big( s^+, F^-_\phi(s^-) \big) \big\rangle,\\
\label{F identity b} \big\langle \phi, p^2 \d_\Sigma^\mp F^\pm_\phi(s^\pm) \big) \big\rangle
&= \big\langle \phi, p^2 \ell_2 \big( h^\mp \d_\Sigma^\mp i^0(\phi), s^\pm \big) \big\rangle,
\end{align}
where $\langle \cdot, \cdot \rangle : \g \otimes \g \to \bbC$
denotes an invariant bilinear pairing on the Lie algebra $\g$.
\begin{proof}
We consider first the identity \eqref{F identity a}. Using \eqref{F beta explicit}
to evaluate $F^\pm_\phi$, \eqref{eqn:Lie} to evaluate $\ell_2$ and \eqref{p def 0} to evaluate $p^2$, we find
\begin{align}
\big\langle \phi, p^2 \ell_2 \big( F^+_\phi(s^+), s^- \big) \big\rangle &= \int_{z \in \bbC} \int_{z' \in \bbC} \omega_{\rm FS}(z) \wedge \omega_{\rm FS}(z') P^{\alpha, \widetilde\alpha}_0(z) P^{\alpha, \widetilde\alpha}_0(z') \frac{\big\langle \phi, \big[[\phi, s^+_0(z')], s^-_0(z) \big] \big\rangle}{z' - z} \notag \\
&= \int_{z' \in \bbC} \int_{z \in \bbC} \omega_{\rm FS}(z') \wedge \omega_{\rm FS}(z) P^{\alpha, \widetilde\alpha}_0(z') P^{\alpha, \widetilde\alpha}_0(z) \frac{\big\langle \phi, \big[s^+_0(z'), [\phi, s^-_0(z)] \big] \big\rangle}{z - z'} \notag \\
&= \big\langle \phi, p^2 \ell_2 \big( s^+, F^-_\phi(s^-) \big) \big\rangle,
\end{align}
where in the second step we used Fubini's theorem and
$\big\langle \phi, \big[[\phi, s^+_0], s^-_0 \big] \big\rangle =
- \big\langle \phi, \big[s^+_0, [\phi, s^-_0] \big] \big\rangle$,
which follows from antisymmetry and Jacobi identity of the Lie bracket
$[\cdot,\cdot]$ on $\g$ and invariance of the bilinear pairing
$\langle \cdot, \cdot \rangle$.
\sk

Consider now the identity \eqref{F identity b}. Using \eqref{F beta explicit}
to evaluate $F_\pm^\phi$, \eqref{d pm} to evaluate $\d_\Sigma^\pm$,
\eqref{hdi phi} to evaluate $h^\mp \d_\Sigma^\mp i_0$ and
\eqref{p def 0} to evaluate $p^2$, we find
\begin{align}
\big\langle \phi, p^2 \d_\Sigma^\mp F^\pm_\phi(s^\pm) \big\rangle &= \int_{z \in \bbC} \int_{z' \in \bbC} \omega_{\rm FS}(z) \wedge \omega_{\rm FS}(z') \, (z \pm 1) \, P^{\alpha, \widetilde\alpha}_0(z') \, \frac{\big\langle \phi, \d_\Sigma^\mp [\phi, s^\pm_0(z')] \big\rangle}{z' - z} \notag\\
&= - \int_{z' \in \bbC} \int_{z \in \bbC} \omega_{\rm FS}(z') \wedge \omega_{\rm FS}(z) \, (z \pm 1) \, P^{\alpha, \widetilde\alpha}_0(z') \, \frac{\big\langle \phi, [\d_\Sigma^\mp \phi, s^0_\pm(z')] \big\rangle}{z - z'} \notag\\
&= \big\langle \phi, p^2 \ell_2 \big( h^\mp \d_\Sigma^\mp i^0(\phi), s^\pm \big) \big\rangle.
\end{align}
In the second step we used Fubini's theorem and
$\big\langle \phi, \big[\phi, [\d_\Sigma^\mp s^\pm_0] \big] \big\rangle = 0$,
which follows from antisymmetry of the Lie bracket $[\cdot,\cdot]$ on $\g$
and invariance of the bilinear pairing $\langle \cdot, \cdot \rangle$.
\end{proof}
\end{lem}

To move forward with our computation of \eqref{eqn:MCaction},
let us recall from \cite[Proposition 3.16]{BSV2} that the transferred
cyclic structure is given by
\begin{align}
\pair{\cdot}{\cdot}_\omega^\prime = \pair{\cdot}{\cdot}_\omega \circ (\widetilde{i} \otimes \widetilde{i}): \F_\cc(\Sigma) \otimes \F_\cc(\Sigma) \longrightarrow \bbC,
\end{align}
where the required compatibility conditions on the original cyclic structure
$\pair{\cdot}{\cdot}_\omega$ from \eqref{cyclic structure FX} are fulfilled as a consequence
of Lemma \ref{lem: compatible pairs}. Let us evaluate the transferred 
cyclic structure on $\phi \in \Omega_\cc^0(\Sigma,\g) = \F_\cc(\Sigma)^1$
and $\phi^\ddagger \in \Omega_\cc^2(\Sigma,\g) = \F_\cc(\Sigma)^2$:
\begin{align}\label{eqn:transcyclic}
\pair{\phi}{\phi^\ddagger}_\omega^\prime &= \frac{\ii}{2 \pi} \int_X \omega \wedge \wedgepair{i^0 \phi + (h^+ \d_\Sigma^+ + h^- \d_\Sigma^-) i^0 \phi}{i^2 \phi^\ddagger} \notag \\
&= \frac{\ii}{2 \pi} \int_\Sigma \int_\bbC (\alpha - \widetilde{\alpha})^2 \, \d z \wedge  \wedgepair{\phi \, \frac{\d \bar{z}}{(1 + |z|^2)^2}}{\phi^\ddagger} \notag \\
&= (\alpha - \widetilde{\alpha})^2 \int_\Sigma \langle \phi, \phi^\ddagger \rangle.
\end{align}
In the first step we used the definition of the original cyclic structure
\eqref{cyclic structure FX} and the formula for $\widetilde{i}$ from \eqref{SDR for field complex}.
In the second step we observed that the term $(h^+ \d_\Sigma^+ + h^- \d_\Sigma^-) i^0 \phi$
does not contribute for degree reasons, we used \eqref{tau in application},
we evaluated $i^2$ using \eqref{i def 0} and we interpreted the meromorphic
$1$-form $\omega$ from \eqref{omega PCM+WZ def} as an $L_{-(\omega)}$-valued
$(1,0)$-form on $\CP$ given by 
$\omega = \big( (\alpha - \widetilde{\alpha})^2 \d z, - (\alpha - \widetilde{\alpha})^2 \d w \big)$.
In the third step we recognized the Fubini-Study volume form $\omega_{\rm FS}$
from \eqref{eqn:FS}, whose integral equals $1$.
\sk

The combination of \eqref{eqn:MCaction} and \eqref{eqn:transcyclic} leads
to the computation of
\begin{align} \label{ell'n term MC}
\langle \phi, \ell'_n(\phi^{\otimes n}) \rangle &= n \, \Big\langle \phi, p^2 (\d^-_\Sigma F^+_\phi + \d^+_\Sigma F^-_\phi) \Big( \widetilde{i}_{n-1}\big( \phi^{\otimes (n-1)} \big) \Big) \Big\rangle \notag\\
&\phantom{=}+ n \, \Big\langle \phi, p^2 \, \ell_2\Big( (h^+ \d^+_\Sigma + h^- \d^-_\Sigma) i^0(\phi), \widetilde{i}_{n-1}\big( \phi^{\otimes (n-1)} \big) \Big) \Big\rangle \notag\\
&\phantom{=}+ \frac{1}{2} \sum_{m=2}^{n-2} \binom{n}{m} \, \Big\langle \phi, p^2 \, \ell_2 \Big( \widetilde{i}_m(\phi^{\otimes m}), \widetilde{i}_{n-m}\big( \phi^{\otimes (n-m)} \big) \Big) \Big\rangle \notag\\
&= (n+1)! \Big\langle \phi, p^2 \, \ell_2\Big( h^+ \d^+_\Sigma i^0(\phi), (F^-_\phi)^{n-2}\big( h^- \d^-_\Sigma i^0(\phi) \big) \Big) \Big\rangle.
\end{align}
In the first step we combined \eqref{elln on phi} with the explicit forms of
$\widetilde{i}$ and $\widetilde{p}$ from \eqref{SDR for field complex}, while
taking into account the graded antisymmetry of $\ell_2$ and the definition
\eqref{linear map F def} of $F^\pm_\phi$.
In the second step we used the explicit expression \eqref{in on phi explicit}
and applied Lemma \ref{lem: F identities} repeatedly.
It remains to compute the right hand side of \eqref{ell'n term MC},
which is the content of Proposition \ref{prop: all MC action terms} below.
We first need a preliminary lemma, which will also be useful in the next section.
\begin{lem} \label{lem: Fn explicit}
For any $n \in \bbZ_{\geq 0}$ and $\phi \in \Omega^0(\Sigma, \g) = \F(\Sigma)^1$ we have
\begin{subequations}\label{Fn explicit}
\begin{align}
(F^\pm_\phi)^n\big( h^\pm \d^\pm_\Sigma i^0(\phi) \big)_0(z) = \frac{1}{(n+1)!} \bigg( \frac{(\widetilde \alpha \mp 1) B(z)^{n+1}}{z - \widetilde\alpha} - \frac{(\alpha \mp 1) \widetilde B(z)^{n+1}}{z - \alpha} \bigg) \, \ad_\phi^n \, \d^\pm_\Sigma \phi,
\end{align}
where
\begin{align}\label{Fn explicit b}
B(z) \coloneqq \frac{(z-\widetilde\alpha) (1+\alpha \bar z)}{(\alpha - \widetilde\alpha) (1+ |z|^2)}, \qquad
\widetilde B(z) \coloneqq B(z) - 1 = \frac{(z-\alpha) (1+\widetilde\alpha \bar z)}{(\alpha - \widetilde\alpha) (1+ |z|^2)}.
\end{align}
\end{subequations}
\begin{proof}
Let $R_n^\pm$ denote the right hand side of \eqref{Fn explicit}. Direct inspection shows that
\begin{equation}
R^\pm_0(z) = \d^\pm_\Sigma \phi \, \frac{\mp \bar z - 1}{1+|z|^2} = \big( h^\pm \d^\pm_\Sigma i^0(\phi)\big)_0(z),
\end{equation}
where the second step is \eqref{hdi phi}, and also
\begin{equation}
\bar\partial R^\pm_n = - \frac{P_0^{\alpha,\widetilde\alpha}(z) \d \bar z}{(1+ |z|^2)^2} [\phi, R^\pm_{n-1}(z)],
\end{equation}
for all $n \geq 1$. Since $R^\pm_n(z) \to 0$ as $|z| \to \infty$,
from Corollary \ref{cor: del bar problem} one has
\begin{equation}
R^\pm_n(z) = - \int_\bbC \omega_{\rm FS}(z') \, P_0^{\alpha,\widetilde\alpha}(z') \, \frac{[\phi, R^\pm_{n-1}(z')]}{z' - z} = F_\pm^\phi \big( R^\pm_{n-1}(z') \big).
\end{equation}
Here the second step follows from the explicit formula \eqref{F beta explicit}.
The claim follows by induction.
\end{proof}
\end{lem}

\begin{propo} \label{prop: all MC action terms}
For $n \in \bbZ_{\geq 1}$ and any $\phi \in \Omega^0(\Sigma, \g) = \F(\Sigma)^1$ we have
\begin{equation} \label{ell'n term MC final}
\langle \phi, \ell'_n(\phi^{\otimes n}) \rangle = 
\begin{cases}
2 \big\langle \phi, \d^+_\Sigma \big( \ad_\phi^{n-1} \d^-_\Sigma \phi \big) \big\rangle, & \text{if $n$ is odd ,}\\
2 \kay \big\langle \phi, \d^+_\Sigma \big( \ad_\phi^{n-1} \d^-_\Sigma \phi \big) \big\rangle, & \text{if $n$ is even},
\end{cases}
\end{equation}
where $\kay$ was defined in \eqref{k def}.
\begin{proof}
The case $n=1$ is immediate from \eqref{eqn:ell'_1}. For $n \geq 2$,
insert \eqref{Fn explicit} into \eqref{ell'n term MC}, use \eqref{hdi phi}
and evaluate $\ell^2$ by \eqref{eqn:Lie} and $p^2$ by \eqref{p def 0 b}. This leads to
\begin{multline}
\langle \phi, \ell'_n(\phi^{\otimes n}) \rangle = -n(n+1) \int_\bbC \frac{(\bar z + 1) P^{\alpha,\widetilde\alpha}_0(z)}{1+|z|^2} \bigg( \frac{(\widetilde \alpha + 1) B(z)^{n-1}}{z - \widetilde\alpha} - \frac{(\alpha + 1) \widetilde B(z)^{n-1}}{z - \alpha} \bigg) \, \omega_{\rm FS}\\
\times \big\langle \phi, \d^+_\Sigma \big( \ad_\phi^{n-1} \d^-_\Sigma \phi \big) \big\rangle,
\end{multline}
where we used also the fact that $\big\langle \phi, [\d^+_\Sigma \phi, \ad_\phi^{n-2} \d^-_\Sigma \phi] \big\rangle 
= \big\langle \phi, \d^+_\Sigma \big( \ad_\phi^{n-1} \d^-_\Sigma \phi \big) \big\rangle$ by the invariance of 
the bilinear pairing $\langle\cdot,\cdot\rangle$ on $\g$.
The integral on the right hand side is computed using the identity
\begin{equation}
-N(N+1) \int_\bbC \frac{P(z') (\bar z' + 1) (1+ z \bar z')^{N-1}}{(1+|z'|^2)^{N}} \omega_{\rm FS}(z') = (z-1) P'(z) - N P(z),
\end{equation}
which holds for any polynomial $P$ of degree at most $N \geq 0$.
(In order to check the latter equation, start from its right hand side
and express the polynomial $P$ via Lemma \ref{lem: reproducing kernel},
recalling also the Fubini-Study volume form $\omega_{\rm FS}$ from \eqref{eqn:FS}.)
The previous identity shows that
\begin{align}
-n(n+1) \int_\bbC \frac{P^{\alpha,\widetilde\alpha}_0(z) (\bar z + 1)}{1+|z|^2} \frac{(\widetilde \alpha + 1) B(z)^{n-1}}{z - \widetilde\alpha}  \omega_{\rm FS} &= \frac{(\widetilde \alpha + 1) (\alpha - 1)}{\alpha - \widetilde\alpha},\\
-n(n+1) \int_\bbC \frac{P^{\alpha,\widetilde\alpha}_0(z) (\bar z + 1)}{1+|z|^2} \frac{(\alpha + 1) \widetilde B(z)^{n-1}}{z - \alpha} \omega_{\rm FS} &= - (-1)^n \frac{(\alpha + 1) (\widetilde \alpha - 1)}{\alpha - \widetilde\alpha},
\end{align}
from which the result follows.
\end{proof}
\end{propo}

We are ready to evaluate the Maurer-Cartan action \eqref{eqn:MCaction} on
a compactly supported degree $1$ element $\phi \in \Omega_\cc^0(\Sigma, \g) = \F_\cc(\Sigma)^1$:
\begin{align}
S_{\rm MC}[\phi] &= (\alpha - \widetilde{\alpha})^2 \sum_{n \geq 1} \frac{1}{(n+1)!} \int_\Sigma \left\langle \phi, \ell'_n(\phi^{\otimes n}) \right\rangle \notag\\
&= \sum_{n \geq 0} \frac{2 (\alpha - \widetilde{\alpha})^2}{(2n+2)!} \int_\Sigma \big\langle \phi, \d^+_\Sigma ( \ad_\phi^{2n} \d^-_\Sigma \phi) \big\rangle + \sum_{n \geq 0} \frac{2 \kay (\alpha - \widetilde{\alpha})^2}{(2n+3)!} \int_\Sigma \big\langle \phi, \d^+_\Sigma (\ad_\phi^{2n+1} \d^-_\Sigma \phi) \big\rangle,\label{MC action full}
\end{align}
where the first step involves \eqref{eqn:transcyclic} and the second step
follows from Proposition \ref{prop: all MC action terms}.
The first term on the right hand side of \eqref{MC action full} can be rewritten as
\begin{align} \label{Action PCM term}
\sum_{n \geq 0} \frac{2 (\alpha - \widetilde{\alpha})^2}{(2n+2)!} \int_\Sigma \big\langle \phi, \d^+_\Sigma ( \ad_\phi^{2n} \d^-_\Sigma \phi) \big\rangle &= - (\alpha - \widetilde{\alpha})^2 \int_\Sigma \bigg\langle \d^+_\Sigma \phi, \sum_{n \geq 0} \frac{2 \, \ad_\phi^{2n}}{(2n+2)!} \d^-_\Sigma \phi \bigg\rangle \notag\\
&= - (\alpha - \widetilde{\alpha})^2 \int_\Sigma \bigg\langle \frac{e^{\ad_\phi} - 1}{\ad_\phi} \d^+_\Sigma \phi, \frac{e^{\ad_\phi} - 1}{\ad_\phi}\d^-_\Sigma \phi \bigg\rangle \notag \\
&= - (\alpha - \widetilde{\alpha})^2 \int_\Sigma \big\langle \d^+_\Sigma g g^{-1}, \d^-_\Sigma g g^{-1} \big\rangle.
\end{align}
In the first step we used Stokes's theorem, in the second step we used
\begin{equation}
\sum_{n \geq 0} \frac{2 \, \ad_\phi^{2n}}{(2n+2)!} = \frac{e^{\ad_\phi} + e^{- \ad_\phi} - 2}{\ad_\phi^2} = \frac{1 - e^{- \ad_\phi}}{\ad_\phi} \frac{e^{\ad_\phi} - 1}{\ad_\phi}
\end{equation}
and invariance of the bilinear pairing $\langle \cdot, \cdot \rangle$ on $\g$, and in the third step
we introduced the group-valued element $g = e^\phi$ and we observed that the
associated right-invariant Maurer-Cartan forms $\d^\pm_\Sigma g g^{-1}$ can be expanded as
\begin{align} \label{dg gI expansion}
\d^\pm_\Sigma g g^{-1} = \sum_{r \geq 1} \frac{1}{r!} \ad_\phi^{r-1} \d^\pm_\Sigma \phi = \frac{e^{\ad_\phi} - 1}{\ad_\phi} \d^\pm_\Sigma \phi.
\end{align}
To rewrite the second term on the right hand side of \eqref{MC action full} we first note that
\begin{equation}
\sum_{n \geq 0} \frac{2 \, \ad_\phi^{2n+1}}{(2n+3)!} = \frac{e^{\ad_\phi} - e^{- \ad_\phi} - 2 \, \ad_\phi}{\ad_\phi^2} = - \int_I \d t \, (e^{- t \, \ad_\phi} - 1) \frac{e^{t \, \ad_\phi} - 1}{\ad_\phi},
\end{equation}
where we have introduced the closed interval $I \coloneqq [0,1]$. We then have
\begin{align}
\sum_{n \geq 0} & \frac{2 \kay (\alpha - \widetilde{\alpha})^2}{(2n+3)!} \int_\Sigma \big\langle \phi, \d^+_\Sigma (\ad_\phi^{2n+1} \d^-_\Sigma \phi) \big\rangle \notag \\
&= \kay (\alpha - \widetilde{\alpha})^2 \int_{\Sigma \times I} \d t \wedge \bigg\langle \ad_\phi \bigg( \frac{e^{t \, \ad_\phi} - 1}{\ad_\phi} \d^+_\Sigma \phi \bigg), \frac{e^{t \, \ad_\phi} - 1}{\ad_\phi} \d^-_\Sigma \phi \bigg\rangle \notag \\
&= \kay (\alpha - \widetilde{\alpha})^2 \int_{\Sigma \times I} \d t \wedge \bigg\langle \phi, \bigg[ \frac{e^{t \, \ad_\phi} - 1}{\ad_\phi} \d^+_\Sigma \phi, \frac{e^{t \, \ad_\phi} - 1}{\ad_\phi} \d^-_\Sigma \phi\bigg] \bigg\rangle \notag \\
&= \kay (\alpha - \widetilde{\alpha})^2 \int_{\Sigma \times I} \big\langle \d_I \hat g \hat g^{-1}, \big[ \d^+_\Sigma  \hat g \hat g^{-1}, \d^-_\Sigma \hat g \hat g^{-1} \big] \big\rangle. \label{eqn:WZterm}
\end{align}
In the first and second steps we used invariance of the bilinear pairing
$\langle \cdot, \cdot \rangle$ on $\g$ and in the third step we introduced the
group-valued path $I \ni t \mapsto \hat g(t) = e^{t \phi}$ from the identity
$\hat{g}(0) = \id$ to $\hat{g}(1) = g$, whose differential fulfills
$\d_I \hat g \hat g^{-1} = \d t \, \phi$.
Inserting \eqref{Action PCM term} and \eqref{eqn:WZterm} into \eqref{MC action full},
one recognizes that the Maurer-Cartan action \eqref{eqn:MCaction} takes the form
\begin{align}
S_{\rm MC}[\phi] = (\alpha - \widetilde{\alpha})^2 \left( - \int_\Sigma \big\langle \d^+_\Sigma g g^{-1}, \d^-_\Sigma g g^{-1} \big\rangle + \kay \int_{\Sigma \times I} \big\langle \d_I \hat g \hat g^{-1}, \big[ \d^+_\Sigma  \hat g \hat g^{-1}, \d^-_\Sigma \hat g \hat g^{-1} \big] \big\rangle \right),
\end{align}
thus reproducing the action of the principal chiral model with a WZ-term.

\subsection{Lax connection} \label{sec:Lax connection}
Let $\phi \in \Omega^0(\Sigma, \g) = \F(\Sigma)^1$ be a Maurer-Cartan element of the
$L_\infty$-algebra $\big( \F(\Sigma), \ell' \big)$. Its image $\lambda_\infty(\phi)$ under the $L_\infty$-morphism
\begin{equation}
\begin{tikzcd}
\lambda_\infty \coloneqq \widetilde{p}_\infty \, \psi \, \widetilde{i}_\infty \; : \; \big( \F(\Sigma), \ell' \big) \arrow[r, smooth squiggle] & \big( \L(\Sigma), \ell' \big),
\end{tikzcd}
\end{equation}
is again a Maurer-Cartan element. Our aim is to identify this element, 
up to gauge equivalence, with the usual Lax connection of the principal chiral model with a WZ-term.
\sk

To this end, consider the degree-$1$ element
\begin{subequations} \label{PCM+WZ Lax}
\begin{equation}
L = L_+ + L_- \in \L(\Sigma)^1 = \bigoplus_{i = \pm} \Omega^i(\Sigma, \g) \otimes \mathcal O_{(\omega)^i_0}(C)
\end{equation}
in $\big( \L(\Sigma), \ell' \big)$, given by the light-cone components
\begin{equation}
L_\pm \coloneqq \frac{\widetilde\alpha \mp 1}{\alpha - \widetilde\alpha} \frac{z - \alpha}{z \mp 1} \d^\pm_\Sigma g g^{-1}
\end{equation}
\end{subequations}
of the Lax connection for the principal chiral model with a WZ-term,
expressed in terms of the formal exponential $g = e^\phi$. We will 
show that 1.)~$L$ is a Maurer-Cartan element of $(\L(\Sigma),\ell')$ and 2.)~$\lambda_\infty(\phi)$ is gauge equivalent to $L$. 
\sk

First, let us argue that both 1.) and 2.) follow under the current hypothesis 
that $\psi \, \widetilde{i}_\infty(\phi)$ is related to $i(L)$ by a gauge transformation. 
For 1.) recall that $\phi$ is a Maurer-Cartan element and $\psi \,\widetilde{i}_\infty$ 
is an $L_\infty$-morphism, hence $\psi\, \widetilde{i}_\infty(\phi)$ is a Maurer-Cartan 
element. Gauge transformations in the differential graded Lie algebra $(\L(X),\ell)$ 
preserve the Maurer-Cartan equation, hence the current hypothesis entails that $i(L)$ 
is a Maurer-Cartan element and, since $i:(\L(\Sigma),\ell^\prime)\to(\L(X),\ell)$ 
is an injective morphism of differential graded Lie algebras, $L$ is a Maurer-Cartan element as well.
For 2.) recall that $\infty$-morphisms between $L_\infty$-algebras induce maps on the 
level of gauge equivalence classes of Maurer-Cartan elements 
\cite[Theorem 5.20 and Proposition 5.23]{KraftSchnitzer}. Since the $\infty$-quasi-isomorphism 
$\widetilde p_\infty$ is a quasi-inverse of $i$, the current hypothesis entails that 
the Maurer-Cartan elements $\lambda_\infty(\phi)$ and $L$ are gauge equivalent.
\sk

Therefore, it remains only to check the hypothesis that 
$\psi \, \widetilde{i}_\infty(\phi)$ is related to $i(L)$ by a gauge transformation.
Applying the quasi-isomorphism $i$ defined in Section \ref{SDR for LX} using \eqref{i def 0 bis}, we find
\begin{equation}
i(L)_0(z) = \sum_{j=\pm} (z-j) L_j(z),
\end{equation}
where $z - j$ for $j = \pm$ is understood to mean $z \mp 1$. Let us also compute $\psi(\widetilde{i}_\infty(\phi))_0(z)$.
By combining Lemma \ref{lem: Fn explicit} and \eqref{in on phi explicit} we find
\begin{equation}
\frac{1}{ n!} \widetilde{i}_n(\phi^{\otimes n})_0(z) = \sum_{j=\pm} \bigg( \frac{\widetilde \alpha - j}{z-\widetilde\alpha} \frac{B(z)^n}{n!} - \frac{\alpha - j}{z-\alpha} \frac{\widetilde B(z)^n}{n!} \bigg) \ad_\phi^{n-1} \d^j_\Sigma \phi.
\end{equation}
Then summing over $n \geq 2$ and introducing $\phi_1 \coloneqq B(z) \phi$ and $\phi_2 \coloneqq \widetilde B(z) \phi$ we have
\begin{equation}
\sum_{n=2}^\infty \frac{1}{n!} \widetilde{i}_n(\phi^{\otimes n})_0(z) = \sum_{j=\pm} \left( \frac{\widetilde \alpha - j}{z-\widetilde\alpha} \sum_{n=2}^\infty \frac{1}{n!} \ad_{\phi_1}^{n-1} \d^j_\Sigma \phi_1 - \frac{\alpha - j}{z-\alpha} \sum_{n=2}^\infty \frac{1}{n!} \ad_{\phi_2}^{n-1} \d^j_\Sigma \phi_2 \right).
\end{equation}
Now recalling also the expression of $\widetilde{i}_1 = \widetilde{i}$ 
from \eqref{tilde i def}, and using \eqref{tau in application} and \eqref{hdi phi}, we also have
\begin{align}
\widetilde{i}_1(\phi)_0(z) = \frac{\phi \, \d\bar z}{(1+|z|^2)^2} - \sum_{j=\pm} \frac{j \bar z + 1}{1+|z|^2} \d^j_\Sigma \phi
= \frac{\phi \, \d\bar z}{(1+|z|^2)^2} + \sum_{j=\pm} \left( \frac{\widetilde\alpha - j}{z - \widetilde\alpha} \d^j_\Sigma \phi_1 - \frac{\alpha - j}{z - \alpha} \d^j_\Sigma \phi_2 \right),
\end{align}
where the second step involves the identity
\begin{align}
\frac{\widetilde{\alpha} \pm 1}{z - \widetilde{\alpha}} B(z) - \frac{\alpha \pm 1}{z - \alpha} \widetilde{B}(z) = \frac{1}{1 + |z|^2},
\end{align}
which is an immediate consequence of \eqref{Fn explicit b}.
Adding the $n=1$ contribution to the above sum over $n \geq 2$ and 
introducing $g_1 = e^{\phi_1}$ and $g_2 = e^{\phi_2}$ we obtain
\begin{align}
\widetilde{i}_\infty(\phi)_0(z) &= \sum_{n=1}^\infty \frac{1}{n!} \widetilde{i}_n(\phi^{\otimes n})_0(z) \notag \\
&= \frac{\phi \, \d\bar z}{(1+|z|^2)^2} + \sum_{j=\pm} \left( \frac{\widetilde \alpha - j}{z-\widetilde\alpha} \frac{e^{\ad_{\phi_1}} - 1}{\ad_{\phi_1}} \d^j_\Sigma \phi_1 - \frac{\alpha - j}{z-\alpha} \frac{e^{\ad_{\phi_2}} - 1}{\ad_{\phi_2}} \d^j_\Sigma \phi_2 \right) \notag \\
&= \frac{\phi \, \d\bar z}{(1+|z|^2)^2} + \sum_{j=\pm} \left( \frac{\widetilde \alpha - j}{z-\widetilde\alpha}\d^j_\Sigma g_1 g_1^{-1} - \frac{\alpha - j}{z-\alpha} \d^j_\Sigma g_2 g_2^{-1} \right),
\end{align}
where for the last step we used also the expansion \eqref{dg gI expansion} of the right-invariant Maurer-Cartan form. 
Let us now also apply the morphism $\psi : \F(X) \to \L(X)$, which is induced from the bundle morphism 
associated with the divisor inequalities arising from the negative divisors entering the 
construction of $\F(X)$, but not of $\L(X)$.
In particular, the evaluation of $\psi$ on $\widetilde{i}_\infty(\phi)$ amounts to 
multiplication by the polynomials in \eqref{P alpha poly def}, therefore we obtain
\begin{equation}
\psi(\widetilde{i}_\infty(\phi))_0(z) = \frac{P_0^{\alpha,\widetilde\alpha}(z) \, \phi \, \d\bar z}{(1+|z|^2)^2} + P_0^{\alpha,\widetilde\alpha}(z) \sum_{j=\pm} \left( \frac{\widetilde \alpha - j}{z-\widetilde\alpha}\d^j_\Sigma g_1 g_1^{-1} - \frac{\alpha - j
}{z-\alpha} \d^j_\Sigma g_2 g_2^{-1} \right).
\end{equation}
Using the fact that $B(z) - \widetilde B(z) = 1$ we can write $g_2^{-1} g_1 = g = e^\phi$ so that
\begin{equation}
\d^\pm_\Sigma g_1 g_1^{-1} = \d^\pm_\Sigma g_2 g_2^{-1} + g_2 \big( \d^\pm_\Sigma g g^{-1} \big) g_2^{-1}.
\end{equation}
Next we observe that we have the following identities
\begin{align}
\frac{P_0^{\alpha,\widetilde\alpha}(z) \, \phi \, \d\bar z}{(1+|z|^2)^2} &= - \bar\partial g_i g_i^{-1},\\
\frac{P_0^{\alpha,\widetilde\alpha}(z)}{z \pm 1} \bigg( \frac{\widetilde \alpha \pm 1}{z-\widetilde\alpha} - \frac{\alpha \pm 1}{z-\alpha} \bigg) &= - 1.
\end{align}
It follows that
\begin{equation} \label{gamma def}
\psi(\widetilde{i}_\infty(\phi))_0(z) = - \bar\partial g_2 g_2^{-1} + \sum_{j=\pm} (z-j) \big( g_2 L_j(z) g_2^{-1} - \d^j_\Sigma g_2 g_2^{-1} \big).
\end{equation}
Recalling from \eqref{d pm Lax} that the contribution $\d_\Sigma$ to the differential 
$\d_\Sigma + \bar\partial$ on $\L(X)$ involves bundle morphisms, we have verified the 
hypothesis that $\psi\,\widetilde{i}_\infty(\phi)$ is related to $i(L)$ by the gauge 
transformation associated with $g_2$. It therefore follows from the above argument that 
$\lambda_\infty(\phi)$ is gauge equivalent to the usual Lax connection $L$ of the 
principal chiral model with a WZ-term, given by \eqref{PCM+WZ Lax}, and that the latter 
is flat as a consequence of $g = e^\phi$ being on-shell.



\section*{Acknowledgments}
The work of M.B.\ is supported in part by the MUR Excellence 
Department Project awarded to Dipartimento di Matematica, 
Universit{\`a} di Genova (CUP D33C23001110001) and it is fostered by 
the National Group of Mathematical Physics (GNFM-INdAM (IT)). 
R.A.C., A.S.\ and B.V.\ gratefully acknowledge the support of the Engineering and Physical
Sciences Research Council (UKRI1723).


%
%


\appendix

\section{\label{sec:proof of SDR}Proofs of Theorems \ref{thm: explicit SDRs 0}, \ref{thm: explicit SDRs -1} and \ref{thm: explicit SDRs -2} and Lemma \ref{lem: compatible pairs}}

We will use two results from complex analysis which we recall below.
\sk

The first result we will need is the weighted Koppelman formula on $\bbC$. 
The standard Koppelman formula is a fundamental result from complex analysis in 
several variables, which provides an integral representation of smooth Dolbeault 
forms on $\bbC^n$, see for instance \cite{Laurent}. In the special case $n=1$, it 
incorporates both the Cauchy-Pompeiu formula when applied to $(0,0)$-forms 
(which itself reduces to the Cauchy integral formula for holomorphic functions) 
and the solution to the $\bar\partial$-problem on $\bbC$ when applied to $(0,1)$-forms. 
The weighted version of the Koppelman formula, see for instance \cite{Gotmark}, 
is a generalization in which the Cauchy kernel is deformed by a smooth weight 
function $W : \bbC^2 \to \bbC$ with the property that its restriction to the 
diagonal is equal to $1$, so that the singularity of the Cauchy kernel is unaffected. 
The weight function is often assumed to be holomorphic in its first argument, but 
we will need also the case when $W$ is only holomorphic in its second argument. 
For completeness, we include the proof of this result.
\begin{theo} \label{thm: weighted Koppelman formula}
Let $R \subset \bbC$ be a bounded open subset and let 
$\zeta \in \Omega^{0,\bullet}(\overline{R})$ be a smooth Dolbeault 
form on the closure $\overline{R}$. Let $W : \bbC^2 \to \bbC$ be a smooth 
function such that $W(z,z) = 1$ for $z \in \bbC$. Then for all $z \in R$ we have the weighted Koppelman formula
\begin{align} \label{weighted Koppelman formula}
\zeta(z) &= \frac{1}{2 \pi \ii} \int_{z' \in \partial R} W(z, z') \frac{\d z' \wedge \zeta(z')}{z' - z} + \frac{1}{2 \pi \ii} \int_{z' \in R} W(z, z') \frac{\d z' \wedge \bar\partial_{z'} \zeta(z')}{z' - z}\\
&\qquad + \frac{1}{2 \pi \ii} \bar\partial_z \int_{z' \in R} W(z, z') \frac{\d z' \wedge \zeta(z')}{z' - z} - \frac{1}{2 \pi \ii} \int_{z' \in R} (\bar\partial_z + \bar\partial_{z'}) W(z, z') \frac{\d z' \wedge \zeta(z')}{z' - z}. \notag
\end{align}
\begin{proof}
Let $\zeta \in \Omega^{0,\bullet}(\overline{R})$ and consider differential form
\begin{equation} \label{eta Koppelman proof}
\eta(z, z') \coloneqq \frac{1}{2 \pi \ii} W(z, z') \frac{\d z' \wedge \zeta(z')}{z' - z} \in \Omega^{1,\bullet}(\overline{R} \times \overline{R} \setminus \text{im}\, \Delta),
\end{equation}
where $\text{im}\, \Delta \subset \overline{R} \times \overline{R}$ 
denotes the image of the diagonal embedding 
$\Delta : \overline{R} \hookrightarrow \overline{R} \times \overline{R}$, 
$z \mapsto (z,z)$. For any $\alpha \in \Omega^{0,\bullet}(\overline{R} \times \overline{R})$ 
the wedge product $\eta(z, z') \wedge \alpha(z, z')$ is locally integrable with respect 
to $z' \in \overline{R}$ so that the fiber integration against \eqref{eta Koppelman proof} 
defines a distribution on $\overline{R} \times \overline{R}$ valued in $\Omega^{0,\bullet}(\overline{R})$.
We then have
\begin{align} \label{del eta Koppelman proof}
(\bar\partial_z + \bar\partial_{z'}) \eta(z, z') &= \frac{1}{2 \pi \ii} (\bar\partial_z + \bar\partial_{z'}) W(z, z') \frac{\d z' \wedge \zeta(z')}{z' - z} - \frac{1}{2 \pi \ii} W(z, z') \frac{\d z' \wedge \bar\partial_{z'} \zeta(z')}{z' - z} \notag\\
&\qquad\qquad + \frac{\ii}{2} W(z, z') \delta(z' - z) (\d \bar z - \d \bar z') \wedge \d z' \wedge \zeta(z'),
\end{align}
where $\delta(z' - z)$ is the Dirac $\delta$-distribution at $z$, with the 
property that for any $\alpha \in \Omega^{0, \bullet}(\overline{R} \times \overline{R})$ 
we have $\int_{z' \in R} \frac{\ii}{2} \delta(z' - z) (\d z - \d \bar z') \wedge \d z'  \wedge \alpha(z, z') 
= (\Delta^\ast \alpha)(z)$. On the other hand, we have
\begin{equation}
\int_{z' \in R} (\bar\partial_z + \bar\partial_{z'}) \eta(z, z') = \bar\partial_z \int_{z' \in R} \eta(z, z') + \int_R \d_{z'} \eta(z, z')
= \bar\partial_z \int_{z' \in R} \eta(z, z') + \int_{z' \in \partial R} \eta(z, z'),
\end{equation}
where in the second step we used Stokes's theorem on the second term. The result 
now follows from substituting \eqref{del eta Koppelman proof} into the integral 
on the left hand side, using the definition of the Dirac $\delta$-distribution 
and the fact that $(\Delta^\ast W)(z) = 1$, and substituting \eqref{eta Koppelman proof} 
into both integrals on the right hand side. 
\end{proof}
\end{theo}

The following corollary will be used to compute certain integrals.
\begin{cor} \label{cor: del bar problem}
For any $\zeta \in \Omega^{0,1}(\bbC)$ such that $\zeta(z) = O(|z|^{-3}) \, \d \bar z$ as $|z| \to \infty$, the integral
\begin{equation}
u(z) = \frac{1}{2 \pi \ii} \int_\bbC \frac{\d z' \wedge \zeta(z')}{z'-z}
\end{equation}
is the unique solution of $\bar\partial u = \zeta$ such that $u(z) \to 0$ as $|z| \to \infty$.
\begin{proof}
The integral defining $u(z)$ converges absolutely since $\zeta(z) = O(|z|^{-3}) \d \bar z$ as $|z| \to \infty$.
\sk

We apply Theorem \ref{thm: weighted Koppelman formula} to the open disc
$R = D_\rho(0)$ of radius $\rho > 0$ around the origin, with the trivial weight
$W(z, z') = 1$ and $\zeta \in \Omega^{0,1}(\bbC)$ restricted to closure
$\overline{R}$. The first, second and fourth terms on the right hand side of
\eqref{weighted Koppelman formula} vanish and in the limit $\rho \to \infty$
we deduce that $\bar\partial_z u(z) = \zeta(z)$.
\sk

To see that $u(z) \to 0$ as $|z| \to \infty$, split the integral over $z' \in \bbC$
into the subregions identified by $|z' - z| < \frac{|z|}{2}$ and $|z' - z| \geq \frac{|z|}{2}$.
The integrals over these subregions are of order $O(|z|^{-2})$ and $O(|z|^{-1})$, respectively, as $|z| \to \infty$.
\sk

Finally, suppose $\tilde u \in \Omega^{0,0}(\bbC)$ is another solution to
$\bar\partial_z \tilde u(z) = \zeta(z)$ such that $\tilde u(z) \to 0$ as $|z| \to \infty$.
Then $u - \tilde u \in \Omega^{0,0}(\bbC)$ is holomorphic and bounded, so constant by Liouville's theorem,
and therefore vanishes since it tends to zero as $|z| \to \infty$.
\end{proof}
\end{cor}

The second complex analysis result that we require is an integral representation 
for polynomials of a bounded degree. The reproducing kernel in this formula will 
serve as motivation for the specific form of the weight function $W$ that we will 
use in applications of Theorem \ref{thm: weighted Koppelman formula}.
\begin{lem} \label{lem: reproducing kernel}
For any polynomial $P$ of degree at most $N \geq 0$ we have
\begin{equation} \label{generic p-integral}
P(z) = - \frac{1+N}{2 \pi \ii} \int_\bbC \frac{P(z') (1+ z \bar z')^N}{(1+|z'|^2)^{N+2}} \d z' \wedge \d \bar z'.
\end{equation}
\begin{proof}
Since $\deg(P) \leq N$ we can write $P(z) = \sum_{k=0}^{N} c_k z^k$. Then expanding 
the numerator of the integrand on the right hand side of \eqref{generic p-integral} in powers of $z$ we find
\begin{equation}
- \frac{1+N}{2 \pi \ii} \int_\bbC \frac{P(z') (1+ z \bar z')^N \d z' \wedge \d \bar z'}{(1+|z'|^2)^{N+2}} 
= - \sum_{k=0}^N \sum_{\ell=0}^N \frac{1}{2 \pi \ii} c_\ell z^k \frac{(N+1)!}{k! (N-k)!} \int_\bbC \frac{z'^\ell \bar z'^k \d z' \wedge \d \bar z'}{(1+|z'|^2)^{N+2}}.
\end{equation}
The integral in the last expression can be computed by moving to polar coordinates 
$z' = r e^{\ii \theta}$. The integral over $\theta$ forces $\ell = k$ in the double 
sum and the integral over $r$ produces a Beta function which is given by a ratio of Gamma functions. Explicitly, we have
\begin{equation}
- \frac{1}{2 \pi \ii} \int_\bbC \frac{z'^\ell \bar z'^k \d z' \wedge \d \bar z'}{(1+|z'|^2)^{N+2}} = \frac{k! (N-k)!}{(N+1)!} \delta_{\ell, k},
\end{equation}
for all $\ell, k \in \{ 0, \ldots, N \}$, from which the result follows.
\end{proof}
\end{lem}

\subsection{Proof of Theorem \ref{thm: explicit SDRs 0}}
We break down the proof of the properties of the strong deformation retract into several lemmas.
\begin{lem} \label{lem SDR pos deg 1}
For any $\alpha \in \bbC^{1+\deg(D)}$ we have $p_D i_D (\alpha) = \alpha$.
\begin{proof}
Given $\alpha = (\alpha_j)_{j=0}^{\deg(D)} \in \bbC^{1+\deg(D)}$, we compute
\begin{equation}
\sum_{j=0}^{\deg(D)} p_D\big( i_D(\alpha) \big)_j z^j = - \frac{1+\deg(D)}{2 \pi \ii} \int_\bbC \frac{\sum_{j=0}^{\deg(D)} \alpha_j z'^j (1 + z \bar z')^{\deg(D)}}{(1 + |z'|^2)^{2+\deg(D)}} \, \d z' \wedge \d \bar z' = \sum_{j=0}^{\deg(D)} \alpha_j z^j,
\end{equation}
where the last equality is by Lemma \ref{lem: reproducing kernel}.
\end{proof}
\end{lem}

\begin{lem} \label{lem SDR pos deg 2}
For any $\zeta \in \Omega^{0,1}(C, L_D)$ we have $- \zeta = \bar\partial h_D(\zeta)$.
Moreover, for any $s \in \Omega^{0,0}(C, L_D)$ we have $i_D p_D(s) = s + h_D(\bar\partial s)$.
\begin{proof}
The claim follows from Theorem \ref{thm: weighted Koppelman formula} 
applied to a disc $R = D_\rho(0) \subset \bbC$ of radius $\rho > 0$ 
around the origin with the weight $W(z, z') = \big( \frac{1+z \bar z'}{1+|z'|^2} \big)^{1+\deg(D)}$, 
after taking the limit $\rho \to \infty$.
\sk

Specifically, let $\zeta = (\zeta_0, \zeta_\infty) \in \Omega^{0,1}(C, L_D)$ and 
take $\zeta$ in Theorem \ref{thm: weighted Koppelman formula} to be the restriction 
$\zeta_0|_{D_\rho(0)} \in \Omega^{0,1}(D_\rho(0))$. The first and second integrals 
on the right hand side of \eqref{weighted Koppelman formula} vanish on degree grounds. 
The third integral tends exactly to $- \bar\partial h_D(\zeta)_0(z)$ as $\rho \to \infty$. 
In the fourth integral we have $\bar\partial_z W(z,z') = 0$ since $W(z,z')$ is holomorphic 
in $z$ and the remaining integrand vanishes on degree grounds. Therefore 
$\zeta_0 = - \bar\partial h_D(\zeta)_0$, and similarly we have 
$\zeta_\infty = - \bar\partial h_D(\zeta)_\infty$, which proves the first statement.
\sk

Now let $s = (s_0, s_\infty) \in \Omega^{0,0}(C, L_D)$ and consider 
$\zeta = s_0|_{D_\rho(0)} \in \Omega^{0,0}(D_\rho(0))$ in Theorem \ref{thm: weighted Koppelman formula}.
Consider the first integral on the right hand side of 
\eqref{weighted Koppelman formula}. Note that as $|z'| \to \infty$, its integrand behaves as
\begin{equation} \label{eta large z'}
W(z, z') \frac{\d z' \wedge s_0(z')}{z' - z} \underset{|z'| \to \infty} \sim s_0(z') \frac{z^{1+\deg(D)}}{z'^{2+\deg(D)}} \d z' = z^{1+\deg(D)} s_\infty(w') \frac{\d z'}{z'^2},
\end{equation}
where in the last step we have used the transformation property 
$s_0(z') = z'^{\deg(D)} s_\infty(w')$ with $w' = z'^{-1}$. Since $s_\infty(w')$ 
is bounded as $|w'| \to 0$, i.e.\ $|z'| \to \infty$, it follows that the first 
integral on the right hand side of \eqref{weighted Koppelman formula} vanishes 
by the estimation lemma. The second integral on the right hand side of 
\eqref{weighted Koppelman formula} tends to $- h_D(\bar\partial s)_0(z)$ as 
$\rho \to \infty$, while the third integral vanishes on degree grounds. As 
for the last integral, since $W(z, z')$ is holomorphic in $z$ the only contribution is
\begin{equation}
- \frac{1}{2 \pi \ii} \int_{z' \in D_\rho(0)} \bar\partial_{z'} W(z, z') \frac{\d z' \wedge s_0(z')}{z' - z} = - \frac{1+\deg(D)}{2\pi \ii} \int_{z' \in D_\rho(0)} \frac{s_0(z') (1+z \bar z')^{\deg(D)}}{(1+|z'|^2)^{2+\deg(D)}} \d z' \wedge \d \bar z',
\end{equation}
which in the limit $\rho \to \infty$ coincides exactly with the expression 
for $i_D p_D(s)_0(z)$. This establishes that $s_0 = - h_D(\bar\partial s)_0 + i_D p_D(s)_0$, 
and similarly we find $s_\infty = - h_D(\bar\partial s)_\infty + i_D p_D(s)_\infty$, 
from which the second statement now follows.
\end{proof}
\end{lem}

\begin{lem} \label{lem SDR pos deg 4}
For any $\zeta \in \Omega^{0,1}(C, L_D)$ we have $p_D h_D(\zeta) = 0$ and $h_D^2(\zeta) = 0$.
Moreover, for any $\alpha \in \bbC^{1+\deg(D)}$ we have $h_D i_D(\alpha) = 0$.
\begin{proof}
The last two facts $h_D^2(\zeta) = 0$ and $h_D i_D(\alpha) = 0$ are both 
trivial on degree grounds. To show also $p_D h_D(\zeta) = 0$,
recall that $\Omega^{0, \bullet}(C, L_D)$ has trivial $1^{\rm st}$ cohomology 
because $\deg(D) \geq 0$. In particular, any $\zeta \in \Omega^{0,1}(C, L_D)$ can be written as
$\zeta = \bar\partial s$ for some $s \in \Omega^{0,0}(C, L_D)$. Taking this into account
and applying the second part of Lemma \ref{lem SDR pos deg 2} and Lemma \ref{lem SDR pos deg 1},
we conclude $p_D h_D (\zeta) = p_D h_D(\bar\partial s) = p_D i_D p_D(s) - p_D(s) = 0$.
\end{proof}
\end{lem}

\subsection{Proof of Theorem \ref{thm: explicit SDRs -1}}
Since in this case the maps $i_D$ and $p_D$ are trivial and $h_D^2 = 0$ 
on degree grounds, to establish the strong deformation retract, the only 
non-trivial property that we must show is 
$\text{id}_{\Omega^{0, \bullet}(C, L_D)} = - \bar\partial h_D - h_D \bar\partial$.
\begin{lem} \label{lem SDR pos deg -1}
For any $\zeta \in \Omega^{0,1}(C, L_D)$ we have $- \zeta = \bar\partial h_D(\zeta)$.
Moreover, for any $s \in \Omega^{0,0}(C, L_D)$ we have $0 = s + h_D(\bar\partial s)$.
\begin{proof}
The proof of Lemma \ref{lem SDR pos deg 2} holds also in the present 
case $\deg(D) = -1$. In particular, note that in this case the weight 
$W(z,z') = 1$ is trivial.
\end{proof}
\end{lem}

\subsection{Proof of Theorem \ref{thm: explicit SDRs -2}}
As in the proof of Theorem \ref{thm: explicit SDRs 0}, we break down the
proof of the properties of the strong deformation retract from Theorem 
\ref{thm: explicit SDRs -2} into several lemmas.
\begin{lem} \label{lem SDR neg deg 1}
For any $\alpha \in \bbC^{-1-\deg(D)}[-1]$ we have $p_D i_D (\alpha) = \alpha$.
\begin{proof}
This check resembles very closely the proof of Lemma \ref{lem SDR pos deg 1} 
by using the complex conjugate of the statement in Lemma \ref{lem: reproducing kernel}.
\end{proof}
\end{lem}

\begin{lem} \label{lem SDR neg deg 2}
For any $\zeta \in \Omega^{0,1}(C, L_D)$ we have $i_D p_D(\zeta) 
= \zeta + \bar\partial h_D(\zeta)$. Moreover, for any $s \in \Omega^{0,0}(C, L_D)$ 
we have $- s = h_D(\bar\partial s)$.
\begin{proof}
The claim follows by the same argument as in the proof of Lemma \ref{lem SDR pos deg 2}, 
however with weight $W(z, z') = \big( \frac{1 + \bar z z'}{1+|z|^2} \big)^{-1-\deg(D)}$.
\end{proof}
\end{lem}

\begin{lem} \label{lem SDR neg deg 4}
For any $\zeta \in \Omega^{0,1}(C, L_D)$ we have $p_D h_D(\zeta) = 0$ and $h_D^2(\zeta) = 0$.
Moreover, for any $\alpha \in \bbC^{-1-\deg(D)}[-1]$ we have $h_D i_D(\alpha) = 0$.
\begin{proof}
The first two facts $p_D h_D(\zeta) = 0$ and $h_D^2(\zeta) = 0$ are both trivial on degree grounds.
To show also $h_D i_D(\alpha) = 0$, we observe that $\bar\partial h_D i_D(\alpha) = i_D p_D i_D(\alpha) - i_D(\alpha) = 0$, 
where we used the first part of Lemma \ref{lem SDR neg deg 2} and Lemma \ref{lem SDR neg deg 1}.
Recalling that $\Omega^{0, \bullet}(C, L_D)$ has trivial $0^{\rm th}$ cohomology 
because $\deg(D) < 0$, we conclude that $h_D i_D(\alpha) = 0$.
\end{proof}
\end{lem}

\subsection{Proof of Lemma \ref{lem: compatible pairs}}
Taking into account the complementarity of the divisors $D + D' = (\omega)$, 
where $\omega \in \mathcal M^{(1)}(C) \setminus \{0\}$ is a given meromorphic $1$-form 
(in particular, $\deg (\omega) = -2$), and the graded symmetry of the wedge product $\wedge$ in the pairing
\eqref{deg -1 pairing}, it suffices to consider two cases, namely 
(a)~$\deg(D) = -1 = \deg(D')$ or (b)~$\deg(D) \geq 0$ and, equivalently, $\deg(D') \leq -2$.

\paragraph{Case (a):} Items (1) and (2) of Lemma \ref{lem: compatible pairs} hold true because
in this case $\text{im}(i_{D}) = 0$ and $\text{im}(i_{D'}) = 0$ are both trivial.
For item (3) the only non-trivial check is when $\zeta \in \Omega^{0,1}(C, L_{D'})$ and
$\alpha \in \Omega^{0,1}(C, L_D)$ have degree $|\zeta| = 1 = |\alpha|$.
We note that, interpreting the meromorphic $1$-form $\omega$ as an $L_{-(\omega)}$-valued $(1,0)$-form on $C$, 
in the chart $z : U_0 \to \bbC$ one has the local presentation $\omega = c \, \dd z$, for some constant $c \in \bbC$.
Then, using the definition \eqref{h def -1} of $h_{D'}$ and $h_D$ and computing the pairing \eqref{deg -1 pairing}
in the chart $z : U_0 \to \bbC$, we have
\begin{align} \label{h zeta alpha 0 1}
\langle\!\langle h_{D'}(\zeta), \alpha \rangle\!\rangle_\omega &= - \frac{c}{(2 \pi)^2} \int_\bbC \dd z \wedge \int_\bbC \frac{\d z' \wedge \zeta_0(z')}{z' - z} \wedge \alpha_0(z) \notag\\
&= \frac{c}{(2 \pi)^2} \int_\bbC \dd z' \wedge \zeta_0(z') \wedge \int_\bbC \frac{\dd z \wedge \alpha_0(z)}{z - z'} = (-1)^{|\zeta|} \langle\!\langle \zeta, h_D(\alpha) \rangle\!\rangle_\omega,
\end{align}
where the second equality follows from Fubini's theorem.

\paragraph{Case (b):} Consider item (1) of Lemma \ref{lem: compatible pairs}. The only non-trivial check is when
$\zeta = i_{D'}(\beta) \in \text{im}(i_{D'})^1$ and $\alpha \in \text{ker}(p_D)^0$. 
It follows from \eqref{p def 0 b} that the condition $p_D(\alpha) = 0$ 
computed using the chart $z : U_0 \to \bbC$ translates into 
\begin{equation}\label{p_D vanishing}
\int_\bbC \frac{\alpha_0(z')\, \bar z'^j}{(1+|z'|^2)^{2+\deg(D)}}\, \d z' \wedge \d \bar z' = 0,
\end{equation}
for all $j = 0, \ldots, \deg(D)$. 
Computing the pairing \eqref{deg -1 pairing} in the chart $z : U_0 \to \bbC$, 
using the local presentation $\omega = c\, \d z$ (see case (a) above),
and recalling the definition of $i_{D'}$ from \eqref{i def -2}, one finds 
\begin{align}
\langle\!\langle \zeta, \alpha \rangle\!\rangle_\omega &= \frac{\ii}{2\pi}\int_\bbC c\, \d z \wedge (1 + |z|^2)^{\deg(D')} \sum_{j=0}^{-2-\deg(D')} \beta_j\, \bar z^j\, \d \bar z \wedge \alpha_0(z) \nn \\
&= \frac{\ii\,c}{2\pi} 
\sum_{j=0}^{\deg(D)} \beta_j \int_\bbC \frac{\alpha_0(z)\, \bar z^j}{(1+|z|^2)^{2+\deg(D)}}\, \d z \wedge \d \bar z = 0,
\end{align}
where we used $\deg(D') = -2 - \deg(D)$ in the second step and \eqref{p_D vanishing} in the last step.
To conclude the proof, we observe that item (2) of Lemma \ref{lem: compatible pairs} is proven by 
a similar argument and item (3) is proven by mimicking case (a).


\end{document}